\documentclass{LMCS}

\def\doi{8(4:9)2012}
\lmcsheading%
{\doi}
{1--28}
{}
{}
{May\phantom.~22, 2012}
{Oct.~22, 2012}
{}

\usepackage{amssymb}
\usepackage{amsmath}
\usepackage{graphicx}
\usepackage{multirow}
\usepackage{xspace}
\usepackage{enumerate,hyperref}

\usepackage{tikz}
\usetikzlibrary{positioning,shapes,petri}

\newcommand{\styletheo}[1]{\sf #1}

\theoremstyle{plain}
\newtheorem{theorem}[thm]{\bf Theorem}

\numberwithin{thm}{section}
\numberwithin{figure}{section}

\newtheorem{proposition}[thm]{\bf Proposition}
\newtheorem{lemma}[thm]{\bf Lemma}
\newtheorem{corollary}[thm]{\bf Corollary}

\theoremstyle{definition}

\newcommand{\cal}{\mathcal}

\newcommand{\insl}{in self-loop\xspace}
\newcommand{\sloop}{$1$-loop\xspace}
\newcommand{\sloops}{$1$-loops\xspace}
\newcommand{\ntrt}{neutral transition\xspace}
\newcommand{\ntrts}{neutral transitions\xspace}

\newcommand{\set}[1]{\{ #1 \}}

\newcommand{\pair}[2]{( #1,#2 )}

\newcommand{\Nat}{\mathbb{N}}
\newcommand{\Zed}{\mathbb{Z}}

\newcommand{\bottom}{\perp}

\newcommand{\card}[1]{{\rm card}(#1)}

\newcommand{\aset}{X}

\newcommand{\asetter}{Z}

\newcommand{\amap}{f}

\newcommand{\egdef}{\stackrel{\mbox{\begin{tiny}def\end{tiny}}}{=}} 
\newcommand{\equivdef}{\stackrel{\mbox{\begin{tiny}def\end{tiny}}}{\equivaut}} 
\newcommand{\equivaut}{\;\Leftrightarrow\;}

\newcommand{\vect}[1]{\vec{#1}}
\newcommand{\avaluation}{\mathbf{v}}

\newcommand{\bo}[1]{\Box}
\newcommand{\diam}[1]{\Diamond}

\newcommand{\avarprop}{{\rm p}}

\newcommand{\aformula}{\varphi}
\newcommand{\aformulabis}{\psi}
\newcommand{\aformulater}{\chi}

\newcommand{\avariable}{{\sf x}}
\newcommand{\avariablebis}{{\sf y}}
\newcommand{\avariableter}{{\sf z}}

\newcommand{\aterm}{t}

\newcommand{\amodel}{\mathcal{M}}

\newcommand{\SOimpl}[1]{\overset{\!\!#1\!\!}{\Longrightarrow}}

\newcommand{\step}[1]{\xrightarrow{\!\!#1\!\!}}
\newcommand{\stepbis}[1]{[#1\rangle}

\newcommand{\pspace}{\textsc{PSpace}}

\newcommand{\expspace}{\textsc{ExpSpace}}

\newcounter{openproblem}
 \def\theopenproblem{\arabic{openproblem}}

\newcommand{\defstyle}[1]{\emph{#1}}
 
 \newcommand{\exesslliaddson}[1]{}

\newenvironment{openproblem}{\refstepcounter{openproblem}\trivlist\item[\hskip\labelsep{\styletheo{Open problem} \theopenproblem}.]}%
                            {\samepage\hfill\mbox{\ $\bigcirc$}\endtrivlist\par}

 \newcounter{example}[section]

\newcounter{exercise}[section]
\def\theexercise{\thesection.\arabic{exercise}}

                           {\samepage\hfill\endtrivlist\par}

\newcommand{\atransition}{t}
\newcommand{\FO}{{\rm FO}}
\newcommand{\FOP}{{\rm FO}^+}
\newcommand{\FOF}{{\rm FO}_{f}}

\newcommand{\ML}{{\rm ML}}
\newcommand{\PAML}{{\rm PAML}}

\newcommand{\aplace}{p}
\newcommand{\amarking}{M}

\newcommand{\interval}[2]{[#1,#2]}

\newcommand{\MC}[2]{{\rm MC}^{{\rm #2}}(#1)}
\newcommand{\VAL}[2]{{\rm VAL}^{{\rm #2}}(#1)}

\newcommand{\apn}{N}
\newcommand{\anetsystem}{N}
\newcommand{\RGV}{{\rm PURG}} 
\newcommand{\RG}[1]{\RGV(#1)}
\newcommand{\URG}[1]{{\rm URG}(#1)}

\newcommand{\UG}[1]{{\rm UG}(#1)}
\newcommand{\brownian}{Br}

\newcommand{\lsuc}{\varphi_{\mathit{l}}}

\newcommand{\reach}[1]{{\rm Reach}(#1)}

\newcommand{\dec}{DEC}
\newcommand{\undec}{UNDEC}

\newcommand{\cut}[1]{}

\titlecomment{{\lsuper*}A short version of this paper was published at FSTTCS 2011.}

\begin{document}

\title[Petri Nets Reachability Graphs: FO Properties]{Petri Net Reachability Graphs:\\
  Decidability Status of First-Order Properties\rsuper*}

\author[P.~Darondeau]{Philippe Darondeau\rsuper a}
\address{{\lsuper a}INRIA Rennes Bretagne Atlantique, Campus de Beaulieu, Rennes, France}	
\email{philippe.darondeau@inria.fr}  

\author[S.~Demri]{Stephane Demri\rsuper b}
\address{{\lsuper b}LSV, ENS Cachan, CNRS, INRIA, France}	
\email{Stephane.Demri@lsv.ens-cachan.fr}  

\author[R.~Meyer]{Roland Meyer\rsuper c}
\address{{\lsuper c}University of Kaiserslautern, Germany}	
\email{meyer@cs.uni-kl.de}  

\author[C.~Morvan]{Christophe Morvan\rsuper d}
\address{{\lsuper d}Universit\'e Paris-Est, Marne-La-Vall\'ee, France}	
\email{christophe.morvan@univ-paris-est.fr}  

\begin{abstract} 
  We investigate the decidability and complexity status of
  model-checking problems on unlabelled reachability graphs of Petri
  nets by considering first-order and modal languages without labels
  on transitions or atomic propositions on markings. We consider
  several parameters to separate decidable problems from undecidable
  ones. Not only are we able to provide precise borders and a
  systematic analysis, but we also demonstrate the robustness of our
  proof techniques.
\end{abstract}

\maketitle

\section{Introduction}
\label{section-introduction}
\paragraph{\em Decision problems for Petri nets.}
Petri nets are among the oldest families of generators of infinite
state systems, and much effort has been dedicated to their algorithmic
analysis.
For Petri nets, the reachability problem is hard but decidable \cite{Mayr84}.
Further important problems that are specific to Petri nets and that were shown decidable are boundedness~\cite{Karp&Miller69,praveen10}, deadlock-freeness and liveness~\cite{Hack1975} (by reduction to reachability),
persistence~\cite{Grabowski80}, and semilinearity~\cite{Hauschildt90}.
Hack's thesis~\cite{Hack1975} provides a comprehensive overview of problems
equivalent to Petri net reachability. On the negative side, language equality
is undecidable for labelled Petri nets~\cite{Hack76,Araki&Kasami76}, but
it can be decided for injectively labelled as well as for labelled and
deterministic Petri nets~\cite{Pelz87} (by a reduction to reachability).
Another undecidability result for Petri nets, obtained by Rabin~\cite{Baker73}
and Hack~\cite{Hack76}, is that equality of reachability sets of two Petri
nets with identical places is undecidable. As our main contribution, we link
this result to first-order logic expressing properties of general Petri net
reachability graphs. We provide a robust proof schema that entails
undecidability of most logical fragments interpreted on such graphs.

\paragraph{\em Our motivations.}
For
Petri nets, model checking CTL formulae with atomic propositions of the form
$p>0$ (place $p$ contains at least one token) is known to be
undecidable~\cite{Esparza98}. This negative result carries over to all fragments
of CTL containing the modalities EF or AF.  Furthermore, model checking CTL
without atomic propositions but with next-time modalities indexed by action
labels is undecidable too~\cite{Esparza98}. In contrast, LTL model-checking over
vector addition systems with states is \expspace-complete~\cite{Habermehl97}
when atomic propositions refer to control states.

These negative results do not compromise the search for
decidable fragments of first-order logic that describe, only purely
graph-theoretically, the shape of the Petri net state graphs. 
So we intentionally avoid edge labels and atomic propositions
interpreted on markings. As an example, we shall consider the
first-order structure $\pair{\Nat^{n}}{\step{}}$ derived from a Petri
net $\apn$ with $n$ places such that $\amarking \step{} \amarking'$
iff $\amarking$ evolves to $\amarking'$ by firing a transition of
$\apn$.
Since $\pair{\Nat^n}{\step{}}$ is an automatic structure,
its first-order theory over predicates $\step{}$ and $=$ is decidable,
see e.g.~\cite{Blumensath&Graedel00}.
This decision procedure
can be extended to Petri net 
state graphs with Presburger-definable
predicates on markings and with labels on transitions.
 As a second example of results related to our work, given a
formula $\aformula$ in $\FO(\step{},=)$ with free variables $\avariable_1,
\ldots, \avariable_{m}$, one can effectively construct a Presburger
formula that characterizes exactly the markings satisfying $\aformula$
in $\pair{\Nat^{n}}{\step{}}$.

However, it is unclear what happens if we consider the first-order
theory of $\step{}$ over the practically interesting structure
$\pair{\reach{\apn}}{\step{}}$. Here, $\reach{\apn}$ denotes the set
of all markings {\em reachable} from the initial marking of Petri net
$\apn$.
Our paper studies this problem. We investigate the decidability status
of several first-order logics, sometimes extended by a bit of MSO (via
reachability predicates), sharing with~\cite{Schulz10} a common
motivation.
The properties of the reachability graph we are interested in  are
\emph{purely graph-theoretical} in that they do not refer to tokens or
transition labels and they are mostly \emph{local} in that we often restrict
ourselves to $\step{}$ instead of its transitive closure. As
summarised in Table~\ref{figure-summary}
(Section~\ref{section-conclusion}) we settle the decidability status
of most problems.  To the best of our knowledge, this is the first
study of logics for the reachability graph. In particular, related
logics in~\cite{Atig&Habermehl10} consider quantitative properties
on markings and transitions, and evaluate formulae on runs. We do not
refer to tokens or to transition labels.

\paragraph{\em Our contributions.}
We investigate the model-checking problem over structures of the form
$\pair{\reach{\apn}}{\step{},\step{*}}$ generated from Petri nets
$\apn$ with first-order languages including predicate symbols for
$\step{}$ and/or $\step{*}$.  We consider variants depending on the
predicates and on whether $\reach{\apn}$ or $\step{*}$ are effectively
semilinear. This allows us to provide a refined analysis about the
decidability borders for such problems.  As it is a classical fragment
of first-order logic, we also consider the modal language
$\ML(\bo{},\bo{}^{-1})$ with forward and backward modalities.
Let us mention some features of our investigation:
\begin{enumerate}[(1)]
\item Undecidability proofs are obtained by reduction from the equality problem
(or the inclusion problem) between reachability sets defined by Petri nets,
shown undecidable in~\cite{Baker73,Hack76}. We demonstrate that our proof schema is
robust and can be adapted to numerous formalisms specifying local properties as
in first-order logic.
Moreover, undecidability can be obtained even for a fixed formula (i.e., for a fixed property).

\item To determine the cause of undecidability, we investigate logical
fragments. At the same time, we strive for maximally expressive decidable
fragments.
With these two goals, our study on graph-theoretical properties is
quite systematic.
\item For decidable problems, we assess the computational complexity ---
either relative to standard complexity classes such as \pspace \ or \expspace
\ or by establishing a reduction from the reachability
problem for
Petri nets (when decision procedures rely on solving instances of
this problem).
\end{enumerate}
Our main findings are as follows
(refined statements can be found in
the body of the paper, see also Table~\ref{figure-summary} in
Section~\ref{section-conclusion}):
\begin{iteMize}{$\star$}
\itemsep 0 cm
\item Model-checking $\pair{\reach{\apn}}{\step{}}$ [resp. $\pair{\reach{\apn}}{\step{*}}$,
      $\pair{\reach{\apn}}{\step{+}}$] is undecidable for the corresponding first-order language
      with a single binary predicate symbol.
\item Undecidability is also shown for the positive fragment of
$\FO(\step{})$, for the
forward fragment of $\FO(\step{})$, and for  $\FO(\step{})$ augmented
with $\step{*}$. The latter result even holds if
the reachability sets are effectively semilinear.
\item
Combining procedures for coverability and reachability in
Petri nets, we obtain some positive results.
We prove that
model-checking the existential fragment of  $\FO(\step{})$ is decidable, but as
hard as the reachability problem for Petri nets.
Moreover, the model checking
problem is decidable for $\FO(\step{},\step{*},=)$ under the  assumption that
the relations $\step{}$ and $\step{*}$ are semilinear (consequence of~\cite{Blumensath&Graedel00}).
We have not found any
decision result between these two extremes.
\item Concerning the modal language  $\ML(\bo{},\bo{}^{-1})$,
      the global model-checking problem on  $\pair{\reach{\apn}}{\step{}}$
is
      undecidable
      but it becomes decidable when restricted
      to $\ML(\bo{})$ (even if extended with Presburger-definable
      predicates on markings); the latter problem is also as hard as
      the reachability problem for Petri nets.
\end{iteMize}

\noindent One may regret that our main results turn towards undecidability but this was
not clear at all when we began our study. On the positive side, we were able to
identify non-trivial fragments for which the decision problems can be of high
computational complexity. Our results shed some new light on the verification
of structural properties on unlabelled net reachability graphs.
\paragraph{\em Structure of the paper.}
The remaining sections are organized as follows.  Section 2 brings the
background of the study. Section 3 presents results that focus on the
reachability graph without the reachability predicate. Section 4
presents those involving the reachability predicate.

\section{Preliminaries}
\label{section-preliminaries}
We recall basics on Petri nets and semilinear sets and
we give the standard definitions and fundamental results used in
the paper. We first introduce the notations needed when considering Petri net
reachability graphs as models for first-order sentences. Then, we define
first-order logic and modal logic interpreted on graphs induced by Petri nets.
Finally, we present positive decidability results about
model-checking problems.
\subsection{Petri nets}
A \defstyle{Petri net} is a bi-partite graph $\apn = (P,T,F, \amarking_0)$,
where $P$ and $T$ are {\em finite} disjoint sets of
\defstyle{places} and \defstyle{transitions}, and $F:(P\times T)\cup(T\times P)
\rightarrow\mathbb{N}$ is a set of directed edges with
non-negative integer weights.
A {\em marking} of $\apn$ is a function $\amarking :P \rightarrow
\Nat$.  $\amarking_0$ is the \defstyle{initial marking} of $\apn$.  A
transition $t \in T$ is \defstyle{enabled at} a marking $\amarking$,
written $\amarking [t\rangle$, if $\amarking(p)\geq F(p,t)$ for all
places $p\in P$. If $t$ is enabled at $\amarking$ then it can
\defstyle{be fired}. This leads to the marking $\amarking'$ defined by
$\amarking'(p)= \amarking(p)+F(t,p){-}F(p,t)$ for all $p\in P$.
The firing relation is denoted by $M[t\rangle M'$.
The definitions are extended to transition sequences $s\in T^*$ in the
expected way.  A marking $M'$ is {\em reachable} from a marking $M$ if
$M[s\rangle M'$ for some $s\in T^*$.  A transition $\atransition$ is
\defstyle{\insl} with a place $\aplace$ iff $F(\aplace,\atransition) =
F(\atransition, \aplace)>0$.  A transition is \defstyle{neutral} if it
has null effect on all places. The \defstyle{reachability set}
$\reach{\apn}$ of $\apn$ is the set of all markings that are reachable
from the initial marking.
\begin{theorem} \cite{Mayr84} \label{theo:mayr}
  Given a Petri net $\apn$ and two markings $\amarking$ and
  $\amarking'$, one can decide whether $\amarking'$ is reachable from
  $\amarking$.
\end{theorem}
\begin{theorem} \cite{Baker73,Hack76} \label{theo-hack}
  Given two Petri nets $\apn$ and $\apn'$, it is undecidable whether
  $\reach{\apn} = \reach{\apn'}$ [resp.  $\reach{\apn} \subseteq
  \reach{\apn'}$].
\end{theorem}
A stronger version of Theorem~\ref{theo-hack} has been established
in~\cite{Jancar95} where it was shown that undecidability still holds
when $\apn$ and $\apn'$ have five places and one of these nets is
fixed.

A Petri net $\apn =(P,T,F,\amarking_0)$ induces several standard
structures on which first-order logics may be interpreted.  
The \defstyle{plain unlabelled reachability graph} of $\apn$ is the
structure $\RG{\apn}=\pair{D}{\step{}}$ where $D= \reach{\apn}$ and
$\step{}$ is the binary relation on $D$ defined by $M\step{}M'$ if
$M[t\rangle M'$ for some $t\in T$. Note that $\amarking_0\in D$ but no
predicate is given to identify this specific marking.
The \defstyle{unlabelled reachability graph} of $\apn$ is the
structure $\URG{\apn}=\pair{D}{init, \step{}, \step{*}, \step{+},=}$
where $init=\set{M_0}$, and
relations $\step{*}$ and $\step{+}$ are the iterative and strictly
iterative closures of $\step{}$, respectively.
The \emph{unlabelled transition graph} of $\apn$ is the structure
$\UG{\apn}=\pair{\mathbb{N}^P}{init,\step{},\step{*}, \step{+},=}$
where $\amarking \step{} \amarking'$ if $\amarking
\stepbis{\atransition} \amarking'$ for some transition $\atransition
\in T$. Note that reachability of markings is not taken into
  account in $\UG{\apn}$.
In the sequel, by default $\card{P} = n$ and we identify $\Nat^{P}$
and $\Nat^{n}$.  We also call \defstyle{\sloop} an edge $\amarking
\step{} \amarking'$ with $\amarking = \amarking'$.
\subsection{Petri nets and semilinear sets}
We rely on results about the semilinear subsets of $\mathbb{N}^n$
that represent possible markings of a Petri net with $n$ places.
Recall that $(\mathbb{N}^n,+)$ is a commutative monoid
where the product operation is the componentwise addition of
$n$-vectors ($+$) and the neutral element is the null $n$-vector.

A subset $E\subseteq\mathbb{N}^n$ is called {\em linear} if it can be
expressed as $x+\{y_1,\ldots,y_m\}^*$ for vectors $x\in\mathbb{N}^n$ and
$y_1,\ldots,y_m\in\mathbb{N}^n$. The Kleene iteration $\{y_1,\ldots,y_m\}^*$ is
a shorthand notation for $k_1y_1+\ldots +k_my_m$ for some $k_1, \ldots,
k_m\in\mathbb{N}$. A subset $E\subseteq\mathbb{N}^n$ is {\em semilinear} if it
is a finite union of linear subsets. Owing to the commutativity of the product
operation $+$, semilinear subsets of $\mathbb{N}^n$ coincide with the regular
subsets of $\mathbb{N}^n$. Hence, they are generated by finite automata over
$\mathbb{N}^n$. Indeed, one can always choose finite automata whose transitions
are labelled with generators, i.e., with $n$-vectors with a single non-null
entry equal to $1$. The semilinear subsets of $\mathbb{N}^n$ form an effective
Boolean algebra \cite{Ginsburg&Spanier64}, hence providing decision procedures
for emptiness. In \cite{Ginsburg&Spanier66}, Ginsburg and Spanier gave an
effective correspondence between semilinear subsets and Presburger subsets,
i.e., subsets of $\mathbb{N}^n$ definable in Presburger arithmetic. Presburger
arithmetic can be decided in triple exponential time~\cite{BoudetC96}.
\begin{proposition} \label{proposition-reach-semilinear}
  Given a Petri net $\apn= (P,T,F,\amarking_0)$ and a semilinear
  subset of markings $E\subseteq\mathbb{N}^{|P|}$, one can decide
  whether (some marking in) $E$ can be reached from $\amarking_0$.
\end{proposition}
Hack reduced this semilinear reachability problem to the reachability
problem in Petri nets \cite[Lemma 4.3]{Hack76}. The proposition now
follows with the decidability of reachability in Theorem~\ref{theo:mayr}.
The statement shows in particular that for any marking $M\in\mathbb{N}^{|P|}$,
one can decide whether a marking greater than or equal to $M$ is
reachable.

We recalled in the introduction that it is decidable whether the reachability
set of a Petri net system is semilinear. Note that semilinearity of the
reachability set $\reach{\apn}$ does not entail semilinearity of the
reachability relation $\step{*}\,\subseteq \reach{\apn} \times  \reach{\apn}
\subseteq \mathbb{N}^{|P|+|P|}$.
Here are some classes of Petri nets and
counter systems for which the reachability relation $\step{*}$ is effectively
semilinear (apart from bounded Petri nets):
\begin{iteMize}{$\star$}
\item Cyclic Petri nets,  see e.g.~\cite{Araki&Kasami77,Bouziane&Finkel97,Leroux11b}.
\item Communication-free Petri nets~\cite{Esparza97}.
\item Vector addition systems with states of dimension 2~\cite{Hopcroft&Pansiot79,Leroux&Sutre04}.
\item Single-path Petri nets~\cite{Howell&Jancar&Rosier93}.
\item Petri nets with regular languages~\cite{Valk&VidalNaquet81}.
\item Flat affine counter systems with the finite monoid
property~\cite{Boigelot98,Finkel&Leroux02b}.
\item Flat relational counter
systems~\cite{Comon&Jurski98,Bozga&Iosif&Konecny10}.
\item Reversal-bounded counter systems~\cite{Ibarra78}.
\end{iteMize}
Some of these results require complex machinery but they are essential
to use the decision procedures based on effective semilinearity.
%
\subsection{First-order languages}
\label{section-preliminaries-fo}
To specify properties of structures $\URG{\apn}$, $\RG{\apn}$ and
$\UG{\apn}$ obtained from a Petri net $\apn$, we introduce a
first-order logic $\FO$ with atomic 
predicates $\avariable \step{} \avariablebis$, $\avariable \step{*}
\avariablebis$, $\avariable \step{+} \avariablebis$ and
$init(\avariable)$.  Formulae in $\FO$ are defined by
$$
\avariable \step{} \avariablebis \ \mid \
\avariable \step{*} \avariablebis \ \mid \
\avariable \step{+} \avariablebis \ \mid \
init(\avariable) \ \mid \
\avariable = \avariablebis
 \ \mid \
\neg \aformula \ \mid \
\aformula \wedge \aformula \ \mid \
\exists \ \avariable \ \aformula
 \ \mid \
\forall \ \avariable \ \aformula.
$$
Given a set $\mathtt{P}$ of predicate symbols from the above
signature, we denote the \emph{restriction of $\FO$ to the predicates
  in $\mathtt{P}$} by $\FO(\mathtt{P})$.  By default, $\FO$ refers to
the full language.  Formulae are interpreted either on $\RG{\apn}$,
$\URG{\anetsystem}$ or $\UG{\anetsystem}$.  Observe that $\FO$ on
$\UG{\anetsystem}$ enables, using $init$ and reachability predicates,
to relativize formulae to $\URG{\anetsystem}$, but restricted logical
languages motivate the existence of both structures. It is worth
noting that by slight abuse, we sometimes use the same notation for a
predicate symbol and its fixed interpretation.
Note that, as regards interpretation, $\step{*}\,=\!(=\cup\step{+})$
and $\step{+}\,=\!(\step{}\circ\step{*})$, hence
$\FO(init,\step{},\step{+},=)$, $\FO(init,\step{},\step{*},=)$, and
$\FO(init,\step{},\step{+},\step{*},=)$ are equally expressive.
$\FO$ indicates that one can quantify over markings. Note that
predicates $\step{+}$ or $\step{*}$ exceed the expressiveness of usual
first-order logics on graphs.
%
We omit the standard definition of the satisfaction relation
$\mathcal{U}, \avaluation \models \aformula$ with $\mathcal{U}$ a
structure ($\RG{\apn}$, $\URG{\apn}$ or $\UG{\apn}$) and $\avaluation$
a valuation of the free variables in $\aformula$.
For example, $\forall\avariable \ \aformula$ holds true whenever the
formula $\aformula$ holds true for all elements (markings) of the
considered structure. \defstyle{Sentences} are closed formulae, i.e.,
without free variables. If $\mathcal{U}\models \aformula$ then
$\mathcal{U}$ is called a model of $\aformula$.

It is worth noting that $\FO$ can only describe
\emph{graph-theoretical} properties of the structures $\mathcal{U}$,
apart from equality tests. The binary relations do not use transitions
of nets as labels and no atomic propositions give reference to
markings. As a consequence, quantitative properties about markings
cannot be expressed in $\FO$, at least in the obvious way, and
constraints about the firing of specific transitions cannot be
expressed either. Note that $\FO$ is not minimal when it comes to
expressiveness.  The redundancies, however, help us design interesting
logical
fragments.\\[0.1cm]
\noindent
In the sequel, we consider several model-checking problems.
The model-checking problem $\MC{\FO}{URG}$ is stated as follows:\vspace{0.1cm}
\begin{description}
\item[\quad input:] a Petri net  $\apn=(P,T,F,M_0)$ and a sentence
$\aformula \in \FO$
\item[\quad question:] $\URG{\apn} \models \aformula$?
\end{description}
\vspace{0.1cm}The variant $\MC{\FO}{UG}$ is:\vspace{0.1cm}
\begin{description}
\item[\quad input:] a Petri net $\apn=(P,T,F,M_0)$ and a sentence
$\aformula \in \FO$
\item[\quad question:] $\UG{\apn} \models \aformula$?
\end{description}
\vspace{0.1cm}
The logics $\FO(\mathtt{P})$ (atomic formulae restricted to
predicates in $\mathtt{P}$) induce restricted variants of the two
model checking problems that we denote by $\MC{\FO(\mathtt{P})}{URG}$ and
$\MC{\FO(\mathtt{P})}{UG}$, respectively.
Formulae in $\FO$ can express standard structural properties, for instance
deadlock-freeness with
$\forall \avariable \ \exists \avariablebis\ \avariable \step{} \avariablebis$,
existence of a \sloop with $\exists
\avariable\ \avariable \step{} \avariable$, or cyclicity with
$\forall \avariable
\forall \avariablebis\ \avariable \step{*} \avariablebis \Rightarrow
\avariablebis \step{*} \avariable$.
Automatic structures form a large class of structures having a
decidable model checking problem for $\FO$.  These structures have
presentations in which $k$-ary relations are defined by synchronous
automata (see~\cite{Blumensath&Graedel00} for more details).
\begin{theorem} \cite{Blumensath&Graedel00}
\label{theo-dec-auto}
  Let $\cal S$ be an automatic structure, then $\MC{\FO}{\cal S}$ is
  decidable.
\end{theorem}
From~\cite{Ginsburg&Spanier64}, semilinear sets and semilinear
relations are automatic. In particular, this means that $(\Nat^n,
\step{}, =)$ is automatic.
Propositions~\ref{proposition-dec-ug},~\ref{proposition-n1}
and~\ref{proposition-dec-reach} are consequences of Theorem~\ref{theo-dec-auto}; they are
provided below to present more explicitly what is the current state of knowledge.
\begin{proposition} \label{proposition-dec-ug}
$\MC{\FO(\step{},=)}{UG}$ is decidable.
\end{proposition}
\noindent
Note that given  $\aformula$ in $\FO(\step{},=)$, one can
effectively build a Presburger formula that characterizes exactly the
valuations satisfying $\aformula$ in $\UG{\apn}$.
Decidability is preserved with Presburger-definable properties on markings and
with labelled transition relations $[t\rangle$.
%
However, having $\Nat^n$ as a domain does not always guarantee
decidability, see the undecidability result in~\cite[Theorem
2]{Schulz10} about a structure with domain $\Nat^n$ but equipped with
successor relations for each dimension and with reachability
predicates constrained by regular languages.
Likewise, subproblems of $\MC{\FO}{URG}$ may require additional
assumptions to achieve decidability, as the semilinearity assumption
made in the statement below. The proposition also follows from
Theorem~\ref{theo-dec-auto}.
\newcommand{\stmntpropnone}{%
Let $\mathcal{C}$ be a class of
Petri nets for which the restriction on reachable markings of the reachability
relation $\avariable \step{*} \avariablebis$ is effectively semilinear. Then,
$\MC{\FO}{URG}$ restricted to $\mathcal{C}$ is decidable.
}
\begin{proposition}
\label{proposition-n1}
\stmntpropnone
\end{proposition}
\begin{proof} 
Let $ \apn =(P,T,F, \amarking_0)$ be a Petri net in $\mathcal{C}$
with $\card{P} = n$. 
We represent its
markings by vectors $\amarking \in\Nat^n$. 
By assumption, $\reach{\apn}$ 
 and the
set $\set{\pair{\amarking}{\amarking'} \,|\, \amarking,\amarking'\in
\reach{\apn} \ {\rm and} \ \amarking\step{*}\amarking'}$ are effectively
semilinear. 
Similarly, the set
$\set{\pair{\amarking}{\amarking} \,|\, \amarking\in
\reach{\apn}}$ is effectively semilinear.  
Define
$\Delta= \set{\pair{\amarking}{\amarking'} \,|\, \amarking,\amarking'\in
\reach{\apn} \ {\rm and} \ \amarking\step{*}\amarking', \ \amarking \neq \amarking'}$.
Then  $\Delta$ is effectively semilinear. Let
$\Delta^2=\set{\pair{\amarking}{\amarking'}\,|\,(\exists \amarking'')\,\pair{\amarking}{
\amarking''} \in\Delta \ {\rm and } \ \pair{\amarking''}{\amarking'} \in\Delta}$. As
semilinear sets are closed under projection (quantifier elimination in
Presburger arithmetic),
$\Delta^2$ is effectively semilinear. Now
$\set{\pair{\amarking}{\amarking'} \,|\, \amarking \in
  \reach{\apn} \ {\rm and} $ \\
  $ \amarking \step{+} \amarking'}$ is equal to
$\Delta\cup\Delta^2$. Hence this set is effectively
semilinear. Therefore, through the effective correspondence between
semilinear sets and sets definable in Presburger arithmetic,
any sentence $\aformula$ of $\FO$ translates to a sentence
$\aformula'$ of Presburger arithmetic logic such that
$\URG{\anetsystem} \models \aformula$ if and only if $\aformula'$ is
true. The proposition follows from the decidability of Presburger
arithmetic~\cite{Presburger29}.
\end{proof}

%
When reachability sets are effectively semilinear but the reachability
relation is not, the strictly less expressive logical fragment
$\FO(\step{},=)$ remains decidable, from Theorem~\ref{theo-dec-auto}.
\newcommand{\stmntpropdecreach}{%
Let $\mathcal{C}$ be a class of
Petri nets $\apn$ for which $\reach{\apn}$ is effectively semilinear.
Then, $\MC{\FO(\step{},=)}{URG}$ restricted to  $\mathcal{C}$ is decidable.
}
\begin{proposition} \label{proposition-dec-reach}
\stmntpropdecreach
\end{proposition}
\begin{proof} 
Consider a Petri net $\apn = (P,T,F,\amarking_0)$ in $\mathcal{C}$. Assume
the Presburger
formula $\aformula(\avariable_1, \ldots, \avariable_{n})$ characterizes
$\reach{\apn}$ where $|P| = n$. There is a second Presburger formula
$\aformula'(\avariable_1, \ldots, \avariable_{n},\avariable_1', \ldots,
\avariable_{n}')$ that characterizes the binary relation $\step{}$ in
$\UG{\apn}$. 

Given a sentence $\aformulabis$ in $\FO(\step{},=)$, one can build  a  sentence
$\amap(\aformulabis)$ in Presburger arithmetic such that $\URG{\apn} \models
\aformulabis$ iff $\amap(\aformulabis)$ is satisfiable in Presburger arithmetic.
The map $\amap(\cdot)$ is homomorphic for Boolean connectives. Furthermore,
\begin{iteMize}{$\star$}
\itemsep 0 cm
\item $\amap(\avariableter \step{} \avariableter') \egdef
       \aformula'(\avariableter_1, \ldots, \avariableter_{n},\avariableter_1', \ldots, 
        \avariableter_{n}')$,
\item $\amap(\avariableter = \avariableter') \egdef
      \underset{i \in \interval{1}{n}}{\bigwedge} \avariableter_i = \avariableter_i'$,
\item $\amap(\forall  \avariableter \ \aformulater) \egdef
       \forall  \avariableter_1, \ldots, \avariableter_{n} \
       (\aformula(\avariableter_1, \ldots, \avariableter_{n})
        \Rightarrow \amap(\aformulater))$.
\end{iteMize}
To evaluate predicate $\step{}$, we resort to 
$\aformula'$.
With
$\aformula$,
we relativize the quantifiers to taking only positions in
$\reach{\apn}$ into account.
\end{proof}

Again, decidability is preserved with Presburger-definable properties on markings and
with labelled transition relations of the form $\step{\atransition}$.
To give an example application of this result, $\MC{\FO(\step{},=)}{URG}$
restricted to cyclic Petri nets is decidable. This follows from
Proposition~\ref{proposition-dec-reach} combined with the fact that cyclic Petri
nets have semilinear reachability sets~\cite{Bouziane&Finkel97}. The restriction
to language $\FO(\step{}, =)$ is essential for the decidability in
Proposition~\ref{proposition-dec-reach}. As we shall see in
Proposition~\ref{proposition-undec-acc-semilinear-sets}, the related model
checking problem $\MC{\FO(\step{}, \step{*})}{URG}$ is undecidable --- even
under the assumption of semilinearity for the reachability sets.
\subsection{Standard first-order fragments: modal languages}
\label{section-preliminaries-ml}
By moving along edges, modal languages provide a local view to
(potentially labelled) graph structures. Note the contrast to first-order
logic in which one quantifies over any element of the structure.
Applications of modal languages include modelling temporal and epistemic
reasoning, and they are central for designing logical specification languages.
In this paper, we consider simple modal languages understood as distinguished
fragments of first-order logic. Moreover, the modal language $\ML$ defined
below has no propositional variable (like  Hennessy-Milner modal
logic~\cite{Hennessy&Milner80} but unlike  standard modal logic K~\cite{Blackburn&deRijke&Venema01})
and no label on modal operators (unlike in  modal languages
dedicated to describing labelled transition systems). This allows us to
interpret modal formulae on directed graphs of the form
$\pair{ \reach{\apn}}{\step{}}$. However, in some places, we shall indicate when
decidability or complexity results can be extended to richer versions of
$\ML$.
The modal formulae in $\ML$ are defined by the grammar
$$
\perp \ \mid \ \top \ \mid \ \neg \aformula \ \mid \ \aformula \wedge \aformulabis \ \mid \
\bo{} \aformula \ \mid \ \diam{} \aformula \ \mid \
\bo{}^{-1} \aformula \ \mid \ \diam{}^{-1} \aformula.
$$
This language is not only poor compared to first-order logic, but also little
expressive compared to other modal dialects. Yet, it is sometimes
sufficiently expressive to obtain first undecidability results for model
checking Petri net structures.
Given a modal formula $\aformula$,
its \defstyle{modal degree} is the greatest number of nested occurrences of modal operators in $\aformula$.
We write $\ML(\bo{})$
to denote the restriction of $\ML$ to the modal operators $\bo{}$ and $\diam{}$.
We interpret modal formulae on directed graphs of the form $\pair{D}{\step{}}$
for some Petri net $\apn=(P,T,F,\amarking_0)$ with $\URG{\apn}=\pair{D}{init,
\step{}, \step{*}, \step{+},=}$.
We provide the definition of the satisfaction relation $\models$
relatively to an arbitrary directed graph $\amodel =
\pair{W}{R}$ (and $w \in W$).
The clauses for Boolean connectives and
logical constants are standard and we omit them. For the modal operators, we
set\vspace{0.1cm}
\begin{iteMize}{$\star$}
\item $\amodel, w\models \bo{} \aformula$
      $\equivdef$ for every $w' \in W$ such that $\pair{w}{w'} \in R$,
      we have $\amodel, w' \models \aformula$.
\item $\amodel, w \models \diam{} \aformula$
      $\equivdef$ there is  $w' \in W$ such that $\pair{w}{w'} \in R$  and
      $\amodel, w' \models \aformula$.
\item $\amodel, w\models \bo{}^{-1} \aformula$
      $\equivdef$ for every $w' \in W$ such that $\pair{w'}{w} \in R$,
      we have $\amodel, w' \models \aformula$.
\item $\amodel, w \models \diam{}^{-1} \aformula$
      $\equivdef$ there is  $w' \in W$ such that $\pair{w'}{w} \in R$  and
      $\amodel, w' \models \aformula$.
\end{iteMize}
\vspace{0.1cm}As usual, $\bo{}$ and $\diam{}$ as well as $\bo{}^{-1}$ and
$\diam{}^{-1}$ are dual operators that can be defined one from another as soon
as negation is part of the language.\\[0.3cm]
The model-checking problem $\MC{\ML}{URG}$  is the
following:
\vspace{0.1cm}
\begin{description}
\item[\quad input:] a Petri net   $\apn=(P,T,F,\amarking_0)$ and a modal formula
$\aformula \in \ML$.
\item[\quad question:] $\pair{Reach(N)}{\step{}}, \amarking_0 \models
\aformula$?
\end{description}
\vspace{0.1cm}
Let $\MC{\ML(\bo{})}{URG}$ denote $\MC{\ML}{URG}$ restricted to
$\ML(\bo{})$. %
Proposition~\ref{proposition-preliminaries-ml1} proves this
model checking problem decidable. The procedure exploits the fact
that a modal formula of modal degree $d$ can only induce constraints on nodes
at distance at most $d$ from the initial marking, a
standard argument, see e.g.~\cite{Blackburn&deRijke&Venema01}.
\newcommand{\stmntpropprelimmlone}{%
$\MC{\ML(\bo{})}{URG}$
is decidable and \pspace-complete.
}
\begin{proposition} \label{proposition-preliminaries-ml1}
\stmntpropprelimmlone
\end{proposition}
\begin{proof}
Consider a Petri net $N=(P,T,F,\amarking_0)$  with $\URG{N}=\pair{D}{init,
\step{}, \step{*}, \step{+},=}$. Let $\aformula$ be a modal formula in
$\ML(\bo{})$ with modal degree $d$ ($d$ is the greatest number of nested
occurrences of modal operators in $\aformula$). We consider the directed graph $\amodel
= \pair{W}{R}$ so that
\begin{iteMize}{$\star$}
\itemsep 0 cm
\item $W \subseteq \Nat^P$ and $R$ is the restriction of $\step{}$ to $W$.
\item For $\amarking \in \Nat^{P}$ we set $\amarking \in W$ $\equivdef$
      there is a sequence of transitions $s$ of length at most $d$ such that
      $\amarking_0 [s \rangle \amarking$.
\end{iteMize}
Observe that $\amodel$ is finite and the cardinal of $W$ is at most exponential
in the size of $N$ and $d$. One can
 show that
$\amodel, \amarking_0 \models \aformula$ iff $\pair{D}{\step{}}, \amarking_0 \models \aformula$.
Hence, $\MC{\ML(\bo{})}{URG}$ is decidable, because the
model-checking problem for $\ML$ over finite structures is decidable (in
polynomial time).
The \pspace \ upper bound can be obtained with an algorithm similar to the one
that shows CTL model-checking over 1-safe Petri nets to be in \pspace,
see e.g.~\cite[Section 4.2]{Esparza98}.
Our problem is actually simpler since we can restrict ourselves
to the temporal operators $\mathtt{A} \mathtt{X}$ and $\mathtt{E} \mathtt{X}$
corresponding to $\bo{}$ and $\diam{}$, respectively.
We briefly describe below the nondeterministic algorithm
$MC((P,T,F,\amarking_0), \aformula)$ that returns $\mathtt{true}$
whenever $\pair{D}{\step{}}, \amarking_0 \models \aformula$. We proceed by a case analysis.
\begin{desCription}
\itemsep 0 cm
\item\noindent{\hskip-12 pt$\aformula = \top$}\ {\bf return} $\mathtt{true}$;
\item\noindent{\hskip-12 pt$\aformula = \neg \aformula'$:}\
     {\bf if} $MC((P,T,F,\amarking_0), \aformula')$
     {\bf then} {\bf return} $\mathtt{false}$ {\bf else} {\bf return} $\mathtt{true}$;
\item\noindent{\hskip-12 pt$\aformula = \aformula_1 \wedge \aformula_2$:}\
     {\bf if} $MC((P,T,F,\amarking_0), \aformula_1)$ {\bf and}
     $MC((P,T,F,\amarking_0), \aformula_2)$
     {\bf then} {\bf return} $\mathtt{true}$ {\bf else} {\bf return} $\mathtt{false}$;
\item\noindent{\hskip-12 pt$\aformula = \bo{} \aformula'$:}\
     {\bf if} for some $\amarking'$ such that $\amarking_0 \step{\atransition}
              \amarking'$ with $\atransition \in T$ we have
     $MC((P,T,F,\amarking'), \aformula') = \mathtt{false}$
     {\bf then} {\bf return} $\mathtt{false}$ {\bf else} {\bf return} $\mathtt{true}$.
\end{desCription}\smallskip

\noindent Note that the depth of recursive calls for $MC((P,T,F,\amarking_0), \aformula)$ is
bounded by the modal degree of $\aformula$ and each call requires only polynomial
space in the size of $(P,T,F,\amarking_0)$ and $\aformula$. Hence,
$MC((P,T,F,\amarking_0), \aformula)$ runs in nondeterministic polynomial space.
By Savitch Theorem, we get the bound \pspace.

To establish \pspace-hardness, we give a reduction from QBF. Let $\mathcal{Q}_1
\avarprop_1 \ \cdots \ \mathcal{Q}_{2n} \avarprop_{2n} \ \aformulabis$ be
a QBF formula where $\mathcal{Q}_1 \cdots \mathcal{Q}_{2n}$ is a sequence of
quantifiers starting with $\mathcal{Q}_1 = \exists$, alternating strictly
$\exists$ and $\forall$, and $\aformulabis$ is a quantifier-free propositional
formula built over the propositional variables in $\set{\avarprop_1, \ldots,
\avarprop_{2n}}$. We consider a modal formula $\aformula$ of the form
$(\diam{}
\bo{})^{n} \aformulabis'$ where $\aformulabis'$ is obtained from
$\aformulabis$ by replacing each propositional variable $\avarprop_i$ by
$\diam{}^{i} \bo{} \perp$. Construct a Petri net $N=(P,T,F,\amarking_0)$ as
follows. The set of places $P$ contains a subset $\set{\aplace_1, \ldots,
\aplace_{2n}}$, in bijection with the atomic propositions and initially empty,
plus auxiliary places. From $\amarking_0$, $N$ executes first a sequence of
$2n$ independent
choices $(t'_1+t''_1)\cdot(t'_2+t''_2)\cdot\ldots\cdot(t'_{2n}+t''_{2n})$ where
$t'_i$ puts $i$ tokens in place $\aplace_i$ to represent the truth of the
corresponding atomic proposition while $t''_i$ puts no tokens in $\aplace_i$ to
indicate the proposition does not hold. After this sequence of
binary choices,  $N$ executes  a non-deterministic
choice  $(x_1+\cdots +x_{2n})$ where $x_i$ removes one token from $p_i$
and puts one token in a place $p'_i$
which was initially empty. Each control place $p'_i$ is set \insl with
a transition $t_i$ that removes at each firing one token from $\aplace_i$.
%
%
  %

Existential quantifications are replaced by $\diam{}$, and universal
ones by $\bo{}$. A path relative to a formula $(\diam{}\bo{})^{n}$
then ends up in a configuration where truth values have been chosen
for all variables.  Note that the formula needs to be true for one
continuation at each $\diam{}$ position and true for each continuation
at $\bo{}$ positions.  The last part of the formula needs to check the
truth values of individual variables.  For each $\avarprop_i$, we have
a formula $\diam{}^{i} \bo{} \perp$ that is true only when there is
precisely a path of length $i$, which corresponds to our encoding of
truth values. The selection of each individual variable (and only one)
is performed by the transition $(x_1+\cdots +x_{2n})$.
  %
Altogether, $\pair{ \reach{\apn}}{\step{}}, \amarking_0 \models (\diam{}
\bo{})^{n}\aformulabis'$ iff $\mathcal{Q}_1 \avarprop_1 \ \cdots \
\mathcal{Q}_{2n} \avarprop_{2n} \aformulabis$ is satisfiable.  Note
that $\reach{\apn}$ is finite.
\end{proof}


For simple
models (like finite structures),
adding $\bo{}^{-1}$ to
$\ML(\bo{})$, often does not change the decidability status or  the
computational complexity of model checking, see
e.g.~\cite{Blackburn&deRijke&Venema01}. When it comes to Petri net
reachability graphs $\RG{\apn}$, adding the backward operator $\bo{}^{-1}$
preserves decidability but at the cost of performing reachability
checks.
\newcommand{\stmntpropprelimmltwo}{%
$\MC{\ML(\bo{},\bo{}^{-1})}{URG}$ is decidable.
}
\begin{proposition} \label{proposition-preliminaries-ml2}
\stmntpropprelimmltwo
\end{proposition}
\begin{proof}
Consider a Petri net $\apn =(P,T,F,\amarking_0)$ with $\URG{N}=\pair{D}{init,
\step{}, \step{*}, \step{+},=}$. Let $\aformula$ be a modal formula in
$\ML(\bo{},\bo{}^{-1})$ of modal degree $d$. Define $\overline{\apn}=(P,T\cup
T^{-1},F,\amarking_0)$ where $T^{-1}$ is a set of formal
inverses of the transitions in $T$, i.e., $F(p,t^{-1})=F(t,p)$ and
$F(t^{-1},p)=F(p,t)$ for all $t\in T$. To model check $\URG{\apn}$ against
$\aformula$, the idea is to consider a depth $d$ unrolling of
$\URG{\overline{\apn}}$. However, when following inverse transitions $\amarking'
[t^{-1}\rangle \amarking$, reachability checks are needed to guarantee the
target marking $\amarking$ belongs to the domain $D$ of structure $\URG{\apn}$.
These checks are effective by Theorem~\ref{theo:mayr}
quoted from~\cite{Mayr84,Kosaraju82,Leroux11}.
More formally, we consider the directed graph $\amodel' = \pair{W'}{R'}$ defined
by\vspace{0.1cm}
\begin{iteMize}{$\star$}
\itemsep 0 cm
\item $W' \subseteq \Nat^P$  and $R'$ is the restriction of $\step{}$ to $W'$.
\item For $\amarking \in \Nat^{P}$ we set $\amarking \in W'$ $\equivdef$
      \begin{enumerate}
      \item $\amarking \in D$,
      \item there is a sequence of transitions $s \in (T \cup T^{-1})^{*}$ of length at
            most $d$ such that $\amarking_0 [s \rangle \amarking$.
      \end{enumerate}
\end{iteMize}
\vspace{0.1cm}Checking $\amarking_0 [s \rangle \amarking$ is easy
whereas $\amarking \in D$ requires a reachability check.
Observe that $\amodel'$ is finite and effectively constructible. The
cardinal of $W'$ is exponential in $d$.
One can
 show that $\amodel', \amarking_0 \models \aformula$ iff
$\pair{D}{\step{}}, \amarking_0 \models \aformula$.
Hence, $\MC{\ML(\bo{},\bo{}^{-1})}{URG}$ is decidable, because
model-checking $\ML$ over finite structures is a decidable problem that
takes polynomial time. 
\end{proof} 

The best known decision procedures for Petri net reachability are
non primitive recursive, which provides the worst possible and
hopefully not tight upper bound to the complexity of the model-checking
problem $\MC{\ML(\bo{},\bo{}^{-1})}{URG}$. Unfortunately, it might well
be the case that this upper complexity bound is tight, for we shall
(in turn) reduce Petri net reachability to the above model-checking
problem in Section~\ref{section-hardness}. 

We introduce another decision problem about $\ML$ that is closely
related to first-order model-checking over reachability graphs.  The
\defstyle{validity problem $\VAL{\ML}{URG}$}, also known as
\emph{global model-checking}, is stated as follows:
\vspace{0.1cm}
\begin{description}
\item[\quad input:]  a Petri net  $\apn=(P,T,F,\amarking_0)$ that
induces the structure $\URG{N}=\pair{D}{init, \step{}, \step{*}, \step{+},=}$,
 and a modal formula $\aformula \in \ML$.
\item[\quad question:]
$\pair{D}{\step{}}, \amarking \models \aformula$
for every marking $\amarking \in D$ ?
\end{description}
\vspace{0.1cm}
As observed earlier, formulae from $\ML(\bo{}, \bo{}^{-1})$ can be
viewed as first-order formulae in $\FO(\step{})$. Therefore, using modal
languages in specifications is a way to consider fragments of
$\FO(\step{})$. Indeed, given a modal formula
$\aformula$ in  $\ML(\bo{},
\bo{}^{-1})$, one can compute in linear time a first-order formula $\aformula'$
with only two individual variables (see
e.g.~\cite{Blackburn&deRijke&Venema01}) that satisfies: for every
Petri net $\apn$ we have $\RG{\apn} \models \aformula'$
iff $\RG{\apn}, \amarking \models \aformula$ for every
marking $\amarking$ in $\reach{\apn}$. Hence,
the validity problem $\VAL{\ML}{URG}$
appears as a natural counterpart to
the model-checking problem for $\FO$
over unlabelled reachability graphs of Petri nets.
We will see in the next section
that both problems are undecidable.

We conclude the section by introducing an extension of $\ML$ that admits
quantifier-free formulae from Presburger arithmetic as atomic propositions. The
idea is to pose arithmetical constraints on the numbers of tokens in places, and
thus to increase the expressiveness of $\ML$. We call this logic $\PAML$ and it
will be mainly used in decidability results in
Section~\ref{section-taming-undecidability}. The domain of the structure for
$\PAML$ needs to be of the form $\Nat^{P}$. More precisely, with terms
$
\aterm ::=  a \times \aplace \ \mid \ \aterm + \aterm
$
where $\aplace$ is a place and $a \in \Zed$ we define $\PAML$ from
$\ML$ by adding atomic formulae $\aformulabis$
defined by $$
\aformulabis ::=
\top \ \mid \
\aterm \leq k \ \mid \
\aterm \geq k \ \mid \
\aterm \equiv_c k' \ \mid \
\aformulabis \wedge \aformulabis
\ \mid \ \neg \aformulabis
.$$
Here, $\top$ is the truth constant, $c \in \Nat \setminus \set{0,1}$,
$k \in \Zed$ and $k' \in \Nat$. The definition of
$\pair{\reach{\apn}}{\amarking}\models \aformulabis$ depends on the definition
of satisfaction of $\aformulabis$ in Presburger arithmetic by a tuple
$\amarking$. The details are as expected and we omit them here.
It can be shown that $\MC{\PAML(\bo{},\bo{}^{-1})}{URG}$ is
decidable. The proof is similar to the proof of
Proposition~\ref{proposition-preliminaries-ml2}.

\section{Structural Properties of Unlabelled Net  Reachability Graphs}
\label{section-URG}

We study the decidability status of model checking unlabelled
reachability graphs of Petri nets against the first-order and modal
logics defined in the previous section.  Recall that the logics are
designed to expressing purely graph-theoretical properties of
reachability graphs.
\subsection{A proof schema for undecidability of $\FO(\step{})$}
\label{section-undec-FO}


To establish undecidability of $\MC{\FO(\step{})}{URG}$, model
checking reachability graphs against first-order specifications, we
provide a reduction of the equality problem for reachability sets. For
two Petri nets $\apn_1$ and $\apn_2$ with identical sets of places,
Hack proved it to be undecidable whether the sets of reachable
markings $\reach{\apn_1}$ and $\reach{\apn_2}$ coincide
(Theorem~\ref{theo-hack} recalls this result from~\cite{Hack76}). To
encode the equality problem into a first-order model checking problem,
we join $\apn_1$ and $\apn_2$ in a third Petri net
$\overline{\apn}$. The construction ensures that equality of the
reachability sets can be checked with a first-order query:
$\reach{\apn_1}= \reach{\apn_2}$ if and only if $\RG{\overline{\apn}}
\models \aformula$. Interestingly, $\aformula$ is a fixed formula and
thus independent of the inputs $\apn_1$ and $\apn_2$.
Before we turn to the technicalities, we sketch the idea of the
construction and comment on why it yields so much expressiveness. With
an initial guess, $\overline{\apn}$ decides to simulate either
$\apn_1$ or $\apn_2$. At any time, $\overline{\apn}$ may stop the
simulation. Then $\overline{\apn}$ either starts behaving in different
ways according to the initial choice between $\apn_1$ and
$\apn_2$. Alternatively, $\overline{\apn}$ may forget this choice and
enter a deadlock marking $M$ that reflects the last marking of
$\apn_1$ or $\apn_2$ in the simulation.

The reachability sets of $\apn_1$ and $\apn_2$ are equal if and only
if every simulation result $M$ can be obtained from both, $\apn_1$ and
$\apn_2$. But inspecting $M$ in isolation does not reveal whether it
stemmed from $\apn_1$ or $\apn_2$. The idea is in the different
behaviours that recall the initial guess when the simulation
ends. They yield a neighbourhood of $M$ in the reachability graph of
$\overline{\apn}$ that reveals the origin of the marking. Indeed, with
\emph{finite experiments} we can check whether $M$ is found in the
simulation of $\apn_1$ or $\apn_2$. Equality of the reachability sets
is then checked by a formula $\aformula$ which requires that, for any
simulation result $M$, both experiments witnessing for $\apn_1$ and
$\apn_2$ succeed. The experiments consist of one backward transition
and some forward transitions. Backward transitions reconstruct the
initial choice, and forward transitions distinguish the nets $\apn_1$
and $\apn_2$.

The strength of this construction stems from the combination of two
ideas. A Petri net can (i) store choices over arbitrarily long
histories and
  (ii) reveal this propagated information in local structures. These
  structures can be characterised by finite back and forth experiments
  that are expressed in terms of first-order formulae.

\vspace{0.1cm}

\noindent\emph{Construction.}\ The two nets $\apn_1$ and $\apn_2$ to be compared
for equality of reachability sets share all places.
The constructed net, $\overline{\apn}$, has these places together with an initialization place $\aplace$, two control places $\aplace_1$ and $\aplace_2$, and additional places $\aplace'_1, \aplace''_1$, and $\aplace'_2$ that we will elaborate on below. The initialization place is the only place that is initially marked, by a single token.

As transitions, $\overline{\apn}$ has the disjoint union of the transitions of $\apn_1$ and $\apn_2$, plus additional transitions that we introduce now together with an explanation of their intended behaviour.
The original transitions are put \insl with the respective control places.
Furthermore, we have two concurrent transitions $\atransition^1_{c}, \atransition^2_{c}$ that consume the initial token and mark either $\aplace_1$ and all
places marked in the initial configuration of $\apn_1$ or $\aplace_2$
and all places marked in the initial configuration of $\apn_2$.
Firing $\atransition^1_{c}$ starts the simulation of $\apn_1$, and similar for $\atransition^2_{c}$.
Each subnet $\apn_1$ and $\apn_2$ may be stopped at any time by firing transitions $\atransition^1_{end}$ and $\atransition^2_{end}$ that move the token from the control place $\aplace_1$ or $\aplace_2$ to the place $\aplace'_1$ or $\aplace'_2$, respectively.
As a result, the token count on the places of $\apn_1$ and $\apn_2$ is not changed any more.

When the transitions $\atransition^1_{end}$ and $\atransition^2_{end}$ have been fired, $\overline{N}$ behaves as indicated in Figure~\ref{GRAPHE1} below $\amarking_1$ and $\amarking_2$, respectively.
At a marking $\amarking_1$, place $\aplace'_1$ enables a transition $\atransition^1_{\ell}$ which puts a token on $\aplace''_1$, depicted by $\amarking_{\circlearrowleft}$ in the figure.
The place enables a transition $\atransition_{sl}$ \insl.
Furthermore, two transitions $\atransition^1_{dl}$ and $\atransition^2_{dl}$ (from $\amarking_1$ to $\amarking_\ell$ and from $\amarking_2$ to $\amarking_r$) empty the
places $\aplace'_1$ and $\aplace'_2$. 
The markings reached by these transitions are designed to be deadlocks.
Moreover, by construction of $\overline{\apn}$,
deadlock markings can only be reached this way (as $\amarking_\ell$ or
$\amarking_r$ or both). 
Since, firing $\atransition^1_{dl}$ or $\atransition^2_{dl}$ lets $\overline{N}$ forget the index $1$ or $2$ of the net that was simulated, we have the following relationship.
Whenever a marking $\amarking$ is reached both in $\apn_1$ and $\apn_2$, the corresponding markings in $\overline{\apn}$ lead to $\amarking_\ell=\amarking_r$

\begin{figure}[htbp]
\begin{center}
{\small 
  %
  \begin{tikzpicture}[scale=.8]
    \tikzstyle{every transition}=[fill, minimum width=5mm, minimum height=2mm];
    \tikzstyle{arc}=[arrows=-latex,auto=left];
    \tikzstyle{multiarc}=[arrows = -latex, bend left, thick, shorten >=2pt];
    \tikzstyle{subnet}=[draw,dashed,shape=regular polygon, regular polygon sides=6];
    
    \draw
    node[subnet](N1){$\apn_1$}
    node[subnet, node distance=3cm, right=of N1](N2){$\apn_2$}  
    node [node distance=1.4cm, right=of N1](tempo){}
    node [node distance=.5cm, above=of tempo](init){$\amarking_0'$}
    node[below=of N1](z1){$\amarking_1$}
    node[below=of N2](z2){$\amarking_2$}
    node[below left=of z1](y){$\amarking_{\circlearrowleft}$}
    node[below right=of z1](zl){$\amarking_\ell$}
    node[below left=of z2](zr){$\amarking_r$}
    (init)edge [arc,swap] node {$\atransition^1_c$} (N1)
    (init)edge [arc] node {$\atransition^2_c$} (N2)
    (N1) edge [dashed, arc,swap] node {$\atransition^1_{end}$} (z1)
    (N2) edge [dashed, arc,swap] node {$\atransition^2_{end}$} (z2)
    (z1) edge [arc,swap] node {$\atransition^1_\ell$} (y)
    (z1) edge [arc] node {$\atransition^1_{dl}$} (zl)
    (z2) edge [arc,swap] node {$\atransition^2_{dl}$} (zr)
    (y) edge[out=160,in=200,loop,arc,swap] node {$\atransition_{sl}$}(y)
    ;
\end{tikzpicture}
}
%

\end{center}
\caption{Reachability graph of $\overline{\apn}$}\label{GRAPHE1}
\end{figure}
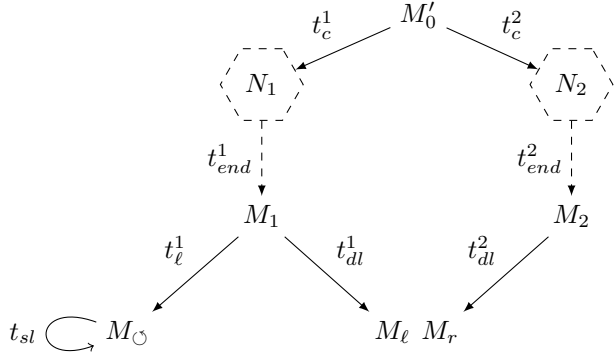
\mbox{}\\
A formula 
expressing equality of the reachability sets of
$\apn_1$ and $\apn_2$ (without recycling variables) is defined hereafter:
$$\aformula \egdef \forall \ \avariableter \ (\neg \exists  \avariableter'\;  \avariableter
\rightarrow  \avariableter')
\Rightarrow (\exists  \avariableter_1\;  \avariableter_1 \rightarrow
\avariableter \wedge
\lsuc( \avariableter_1)) \wedge (\exists  \avariableter_2\;
\avariableter_2\rightarrow  \avariableter
\wedge \neg \lsuc( \avariableter_2))
$$
%
Formula $\lsuc(\avariable) \egdef \exists \ \avariablebis \  ( \avariable
\step{}
\avariablebis \wedge  \avariablebis
\step{} \avariablebis)$ indicates that $\avariable$ has a successor that has a \sloop.
\newcommand{\stmntlemunschem}{
$\reach{\apn_1} = \reach{\apn_2}$ if and only if $\RG{\overline{\apn}} \models
\aformula$.
}
\begin{lemma} \label{lemma-undecidability-schema}
\stmntlemunschem
\end{lemma}
\begin{proof}
  For the implication from left to right, consider a deadlock
  $\amarking$.  Marking $\amarking$ is reachable only via
  $\atransition^1_{dl}$ or $\atransition^2_{dl}$, say $\amarking
  _1\stepbis{\atransition^1_{dl}} \amarking$.  Then marking
  $\amarking_1$ satisfies $\lsuc$ and stems from a marking
  $\amarking_1' \stepbis{\atransition^1_{end}} \amarking_1$ of
  $\apn_1$. The hypothesis on equal reachability sets yields a marking
  $\amarking_2'$ of $\apn_2$ that leads by transition
  $\atransition^2_{end}$ to a marking $\amarking_2$ satisfying $\neg
  \lsuc$ as required.

  In turn, if $\aformula$ holds we establish two inclusions.  To show
  $\reach{\apn_1} \subseteq \reach{\apn_2}$, consider marking
  $\amarking'_1$ reachable via sequence $s_1$ in $\apn_1$. In
  $\overline{\apn}$, the marking can be prolonged to a deadlock $M$
  with $\amarking_0'\stepbis{\atransition_c^1} \amarking_0^1
  \stepbis{s_1} \amarking_1' \stepbis{\atransition_{end}^1 }
  \amarking_1 \stepbis{\atransition_{dl } ^1 } \amarking$. Here,
  $\amarking_1$ satisfies $\lsuc$. But $\aformula$ yields another
  predecessor $\amarking_2$ of $\amarking$ with $\amarking_2 \neq
  \amarking_1$. To avoid the \sloop, marking $\amarking_2$ has to
  result from a sequence $\amarking_0' \stepbis{\atransition_c^2}
  \amarking_0^2\stepbis{s_2} \amarking_2'
  \stepbis{\atransition_{end}^2 } \amarking_2
  \stepbis{\atransition_{dl } ^2 } \amarking$. It is readily checked
  that $\amarking_1'$ and $\amarking_2'$ coincide up to the token on
  the control place. Hence, $\amarking_1' \in \reach{\apn_2}$ as
  required.
\end{proof}
\begin{corollary} \label{prop:undecFO}
  $\MC{\FO(\step{})}{URG}$ is undecidable, already for the fixed
  formula $\aformula$ given in this section.
\end{corollary}
By recycling variables in $\aformula$ above, we get a sharp result that marks
the undecidability border of model checking against $\FO(\step{})$ by two
variables. Model checking $\FO(\step{})$ restricted to one variable is
decidable.
\newcommand{\stmntpropundecFOthree}{
There exists a formula $\aformula$ in $\FO(\step{})$ with two individual variables such that
$\MC{\FO(\step{})}{URG}$ restricted to  $\aformula$  is undecidable.
}

\begin{theorem} \label{prop:undecFO3}
\stmntpropundecFOthree
\end{theorem}
\begin{proof} It is sufficient to observe that formula $\aformula$ below
$$
\forall  \avariableter \ (\neg \exists  \avariableter'\;  \avariableter
\rightarrow  \avariableter') 
\Rightarrow (\exists  \avariableter_1\;  \avariableter_1 \rightarrow 
\avariableter \wedge 
\lsuc( \avariableter_1)) \wedge (\exists  \avariableter_2\; 
\avariableter_2\rightarrow  \avariableter 
\wedge \neg \lsuc( \avariableter_2))
$$
with 
$\lsuc(\avariable) \egdef \exists \avariablebis \  ( \avariable
\step{}
\avariablebis \wedge  \avariablebis                                          
\step{} \avariablebis)$ is logically equivalent to the formula
$$
\forall  \avariableter \ (\neg \exists  \avariableter'\;  \avariableter
\rightarrow  \avariableter') 
\Rightarrow (\exists  \avariableter'\;  \avariableter' \rightarrow 
\avariableter \wedge 
\lsuc'( \avariableter')) \wedge (\exists  \avariableter'\; 
\avariableter'\rightarrow  \avariableter 
\wedge \neg \lsuc'( \avariableter'))
$$
where 
$\lsuc'(\avariableter') \egdef \exists \ \avariableter \  ( \avariableter'
\step{}
\avariableter \wedge  \avariableter                                         
\step{} \avariableter)$. Recycling of variables is explained
e.g. in~\cite{Gabbay81}.
\end{proof}


Moreover, combined with the fact that model checking first order logic
for automatic structures is decidable, Theorem~\ref{prop:undecFO3}
leads to the following impossibility result.
\begin{corollary}\label{coro:notauto}
  There is no algorithm to construct an automatic graph isomorphic to
  the unlabelled reachability graph of a Petri net.
\end{corollary}
%
%
%
%
Note that this negative result cannot follow directly from
complexity-theoretic considerations.
Indeed, even if the unlabelled reachability graph of a Petri net could be
represented as an automatic graph, this automatic graph could not be
used to decide on reachability of markings unless this representation
were in effective bijection with $\Nat^n$ (where $n$ is the number of
places).

%
%

Restricted to a single variable, model checking $\FO(\step{})$ becomes
decidable.
\newcommand{\stmntpropdecfoone}{
$\MC{\FO(\step{})}{URG}$ restricted to one individual variable is
decidable.
}
\begin{proposition} \label{proposition-decidability-fo1}
\stmntpropdecfoone
\end{proposition}
\begin{proof}
Every sentence in $\FO(\step{})$ restricted to one individual variable 
is logically equivalent either to $\perp$, or to $\top$, or to a positive Boolean
formula with atomic formulae of one of the forms below:
\begin{enumerate}[(1)]
\itemsep 0 cm
\item $\exists  \avariable \  (\avariable \step{} \avariable)$
\item $\exists  \avariable \  \neg (\avariable \step{} \avariable)$
\item $\forall  \avariable \ (\avariable \step{} \avariable)$
\item $\forall  \avariable \ \neg (\avariable \step{} \avariable)$.
\end{enumerate}
Since (2) is the negation of (3) and (1) is the negation of (4), decidability is
obtained by evaluating 
(1) $\RG{\apn} \models \exists  \avariable \  (\avariable \step{} \avariable)$ 
and (3) $\RG{\apn} \models \forall  \avariable \ (\avariable \step{}
\avariable)$. (1) can be checked by solving one instance of the covering
problem for each \ntrt of the net whereas (3) can be checked by
solving a single instance of the reachability problem.
 Indeed, let
$T_{\circlearrowleft}$ be the subset of transitions of the net that leave
markings unchanged (neutral transitions). Then the set of markings specified hereafter is  
effectively semilinear:
$$
\asetter \egdef
\set{\amarking: \ {\rm not} \ \amarking [\atransition \rangle  \ {\rm for \ all} \ \atransition \in 
T_{\circlearrowleft}}
$$
We have {\em not} $\RG{\apn} \models \forall \ \avariable \ (\avariable \step{} \avariable)$
iff there is a marking $\amarking \in \asetter$ that is reachable, $\amarking_0
\step{*} \amarking$. With~\cite[Lemma 4.3]{Hack76} this reduces to an instance
of the reachability problem. 
\end{proof}

It is possible to play further with parameters. For instance, our undecidability proof
uses several reachability graphs with constant formulae. It is open whether
there is a fixed Petri net reachability graph for which
the model-checking problem for $\FO(\step{})$ is undecidable.
\subsection{Robustness of the proof schema}
Based on the previous proof schema, this section presents
undecidability results for subproblems of $\MC{\FO(\step{})}{URG}$.
More specifically, we consider the positive fragment, the forward
fragment, the restriction when the direction of edges is omitted, and
$\ML(\bo{},\bo{}^{-1})$.  For all these fragments, we establish
undecidability of model checking.
\subsubsection{Forgetting orientation}
Let $\lambda(\avariable, \avariable') \egdef (\avariable \step{}
\avariable') \vee (\avariable' \step{} \avariable)$.  Expressing
properties about $\RG{\apn}$ in $\FO(\lambda)$ amounts to getting rid
of the direction of edges of this graph.  Despite this weakening,
undecidability is still present for general Petri nets.
To instantiate the above argumentation, we have to identify deadlock
markings and analyse their environment.
In $\FO(\lambda)$, we augment markings encountered during the
simulation by $3$-cycles.  Then, the absence of $3$-cycles and an
environment without such cycles characterises deadlock markings.
\newcommand{\stmntpropundeclambda}{
$\MC{\FO(\lambda)}{URG}$ is undecidable.
}
\begin{proposition}\label{prop-undec-lambda}
\stmntpropundeclambda
\end{proposition}
\begin{proof} 
We take advantage of the fact that $\FO(\lambda)$ can express that a node $x$
belongs to an undirected cycle of length three. A possible formula is:
\[
3cycle(\avariable) \egdef \exists \avariablebis \exists \avariableter\
(\lambda(\avariable,\avariablebis)\wedge\lambda(\avariablebis,\avariableter)\wedge\lambda(\avariableter,\avariable))\wedge
\neg(\lambda(\avariable,\avariable)\vee\lambda(\avariablebis,\avariablebis)\vee\lambda(\avariableter,\avariableter))
\]
Now consider two Petri nets $N_1$ and $N_2$ with identical sets of places. 
For $1\leq i\leq 3$, add to each net new places $p_i$ and transitions $t_i$ such
that $p_1$ contains initially one token, while $p_2$ and $p_3$ are empty.
Transition $t_i$ takes one token from $p_i$ and puts one token in $p_{i+1\mod
3}$. The resulting Petri nets have identical reachability sets if and only if
$N_1$ and $N_2$ have identical reachability sets. Therefore, equality of
reachability sets is undecidable for nets in which every reachable marking
belongs to some cycle of length three. Assuming that $N_1$ and $N_2$ have this
property, let $\overline{N}$ be the net constructed from $N_1$ and $N_2$ as
in 
the proof of 
Proposition~\ref{prop:undecFO3} 
(see also Figure~\ref{GRAPHE1}).
We can assume without loss of generality that every transition of $\apn_1$ and 
$\apn_2$ changes the current
marking (the other transitions do not affect the reachability sets and can be
removed). As a consequence, the reachability 
graphs of the augmented nets $N_1$ and $N_2$ have no \sloops, 
which is
required for the effectiveness of $3cycle(\avariable)$. The deadlock markings of
$\overline{\apn}$ are then exactly the markings that have no cycle of length 
one or three
and that are surrounded by nodes without 
cycles of length three: 
\[
dead(\avariableter)  \egdef \neg\lambda(\avariableter,\avariableter)\ \wedge\ 
\neg 3cycle(\avariableter) \wedge\
\forall \avariable\ \lambda(\avariableter,\avariable) \Rightarrow \neg
3cycle(\avariable).
\]  
Equality of the reachability sets of $N_1$ and $N_2$ is then expressed by
the formula $\aformula$ below
\[
\forall \avariableter\ dead(\avariableter)\Rightarrow (\exists \avariableter_1\ 
\lambda(\avariableter,\avariableter_1)\wedge \varphi_l(\avariableter_1))
\wedge (\exists \avariableter_2\ \lambda(\avariableter,\avariableter_2)\wedge \neg 
\varphi_l(\avariableter_2))
\]
where 
$\varphi_l(\avariableter) \egdef \exists \avariablebis\
\lambda(\avariableter,\avariablebis)\wedge 
\lambda(\avariablebis,\avariablebis)$.
We have $\reach{\apn_1} = \reach{\apn_2}$ iff
$\overline{\apn} \models \aformula$. 
By Theorem~\ref{theo-hack}, $\MC{\FO(\lambda)}{URG}$ is undecidable.
\end{proof}


\subsubsection{A well-known first-order fragment:  $\ML(\bo{},\bo{}^{-1})$}
\label{section-undecidability-ml}
To establish undecidability of the problem
$\VAL{\ML(\bo{},\bo{}^{-1})}{URG}$, we again provide a reduction from
the equality problem for Petri net reachability sets.
\newcommand{\stmntpropundecidabilityml}{
$\VAL{\ML(\bo{},\bo{}^{-1})}{URG}$ is undecidable.
}
\begin{proposition} \label{proposition-undecidability-ml}
\stmntpropundecidabilityml
\end{proposition}
\proof
Consider two Petri nets $\apn_1$ and $\apn_2$ with identical
sets of places. We rely on the construction of $\overline{\apn}$ in
Section~\ref{section-undec-FO}, but give a modal formula $\aformula$
(independent of $\apn_1$ and $\apn_2$) that yields the following
equivalence: $\apn_1$ and $\apn_2$ have the same reachability set  iff
$\RG{\overline{\apn}}, \amarking \models \aformula$ for every marking
$\amarking$ in $\reach{\overline{\apn}}$. 
For all deadlocks, there is one predecessor (from $\apn_1$)
that is able to do two more steps and another predecessor (from $\apn_2$) that
is not:
$\aformula \egdef 
\bo{} \perp\ \Rightarrow\ 
(\diam{}^{-1}\diam{}\diam{}\top\ \wedge\  \diam{}^{-1}\bo{}\bo{}\perp)
$. 
Formula $\aformula$  is semantically equivalent to
the first-order formula $\aformula_{fo}$ defined below:
\begin{align*}
\forall  \avariableter \ (\neg \exists  \avariableter'\;  \avariableter
\rightarrow  \avariableter')\
\Rightarrow&\ (\exists  \avariableter_1, \avariableter_2, \avariableter_3\; 
(\avariableter_1 \step{} \avariableter) \wedge 
(\avariableter_1 \step{} \avariableter_2) \wedge  (\avariableter_2 \step{}
\avariableter_3))\ \wedge \\ 
&\ (\exists  \avariableter_1\; 
(\avariableter_1 \step{}  \avariableter) 
\wedge  \forall  \avariableter_2, \avariableter_3\; \neg ((\avariableter_1 \step{} \avariableter_2) 
\wedge  (\avariableter_2 \step{} \avariableter_3))).
\rlap{\hbox to 52 pt{\hfill\qEd}}
\end{align*}\smallskip

\noindent This undecidability result is tight.  In
Section~\ref{section-dec-val-ml}, we establish decidability of an
extended variant of $\VAL{\ML(\bo{})}{URG}$ where the backward
modality $\bo{}^{-1}$ is excluded.  Moreover, by translating formulae
in $\ML(\bo{},\bo{}^{-1})$ to $\FO(\step{})$ restricted to two
individual variables, we get another evidence that
$\MC{\FO(\step{})}{URG}$ restricted to two individual variables is
undecidable.
\subsubsection{$\FO(\step{})$ restricted to positive or forward formulae}
Although $\VAL{\ML(\bo{},\bo{}^{-1})}{URG}$ and
$\MC{\FO(\step{})}{URG}$ are undecidable in general, we have
identified decidable fragments of modal logic in
Section~\ref{section-preliminaries-ml}. By analogy, one may expect to
find decidability of related fragments of first-order logic. We prove
here that this is not the case. We consider forward $\FO(\step{})$ and
positive $\FO(\step{})$ and show that their model checking problems
are undecidable.  In a \defstyle{positive formula}, atomic
propositions occur only under the scope of an even number of
negations. Let ${\FOP( \mathtt{P} )}$ denote the set of positive first
order formulae over predicates in $\mathtt{P}$.
\newcommand{\stmntproppositive}{
$\MC{\FOP(\step{})}{URG}$ is undecidable.
}

\begin{proposition}\label{prop-positive}
\stmntproppositive
\end{proposition}
\begin{proof}
  We rely on the previously introduced proof schema. Let $\apn_1$ and
  $\apn_2$ be two Petri nets and $\overline{\apn}$ their combination
  sketched in Figure~\ref{GRAPHE1}.  We propose a positive formula
  $\aformula$ so that inclusion $\reach{\apn_2} \subseteq
  \reach{\apn_1}$ holds if and only if $\RG{\overline{\apn}}\models
  \aformula$:
  \[
  \aformula\ \egdef\ \forall \avariableter\ \exists
\avariableter_1\ \exists \avariablebis_\ell\, \exists
\avariableter' \,
 (\avariableter \step{} \avariableter') \vee ((\avariableter_1\step{}
\avariableter)\wedge(\avariableter_1\step{}
\avariablebis_\ell)\wedge(\avariablebis_\ell\step{} \avariablebis_\ell))
  \]
The formula considers an arbitrary marking $\amarking$. If $\amarking$
is no deadlock, nothing is required by $\aformula$. If $\amarking$
is a deadlock, then $\aformula$ asks for vertices
$\amarking_1$ and $\amarking_\circlearrowleft$ so that $\amarking_1$ is a common
direct ancestor of $\amarking$ and $\amarking_\circlearrowleft$ and moreover
$\amarking_\circlearrowleft$ has a \sloop.  

By construction of $\overline{\apn}$, formula $\aformula$ is satisfied
if and only if every deadlock marking $\amarking$ reachable in
$\overline{\apn}$ (in particular, a simulation of $\apn_2$) can be
reached in $\apn_1$. This means $\reach{\apn_2} \subseteq
\reach{\apn_1}$. 
\end{proof}


\begin{openproblem}
Decidability status of $\MC{\FOP(\step{*})}{URG}$.
\end{openproblem}
A {\em forward formula} is a formula in which every occurrence $\avariable
\step{} \avariablebis$ is in the scope of a quantifier sequence of the form
$Q_1 \ \avariable \ldots Q_2\ \avariablebis$ where $\avariable$ is bound
before $\avariablebis$. Let ${\FOF( \mathtt{P} )}$ denote the set of forward
formulae over predicates in $\mathtt{P}$.
\newcommand{\stmntpropforward}{
$\MC{\FOF(\step{})}{URG}$ is undecidable.
}
\begin{proposition}\label{prop-forward}
\stmntpropforward%
\end{proposition}%
\begin{proof}
We again reduce the equality problem for reachability sets of two Petri nets
$\apn_1$ and $\apn_2$. Let $\overline{\apn}$ be the net presented in
Figure~\ref{GRAPHE1}. We propose a forward formula $\aformula$ so that
$\reach{\apn_2} =  \reach{\apn_1}$ if and only if $\RG{\overline{\apn}}\models
\aformula$:
  \[
  \aformula \egdef\ \forall \avariableter_2\ \exists \avariableter_1\ \forall
\avariableter\ \exists \ \avariablebis_\ell\ \ \exists \avariableter'\ 
  (\avariableter_2 \step{} \avariableter) \Rightarrow  ( 
  (\avariableter \step{} \avariableter')\vee
  \aformulabis(\avariableter_1,\avariableter_2,\avariableter,\avariablebis_\ell)
%
  )
\]
\\[-0.8cm]
\[
\aformulabis(\avariableter_1,\avariableter_2,\avariableter,
\avariablebis_\ell)\egdef (\avariableter_1\step{}
  \avariableter)
\wedge(\avariablebis_\ell\step{}
  \avariablebis_\ell)\ 
  \wedge\ ((\avariableter_1\step{}\avariablebis_\ell)\Leftrightarrow
	\neg(\avariableter_2\step{}\avariablebis_\ell))
\]
Forward formulae make it harder to quantify over deadlock markings
$\amarking$. 
Before presenting how formula $\aformula$ enables the reduction, a 
short comment on quantification: this formula intends to quantify 
over $\avariableter$, but the forward constraint imposes first to
quantify over $\avariableter_2$, then on $\avariableter_1$, and 
only afterwards on $\avariableter$. This is not a problem since, once 
$\avariableter_2$ is fixed, variable $\avariableter_1$ may be fixed,
and then $\avariableter$ may be chosen.
The idea of $\aformula$ is to capture the situation in
Figure~\ref{GRAPHE1}, potentially with the roles of $M_1$ and $M_2$
swapped. In detail, the formula considers an arbitrary marking
$\amarking_2$, a corresponding marking $\amarking_1$ (if it exists),
and an arbitrary marking $\amarking$.  If $\amarking_2$ and
$\amarking$ are not connected, then $\aformula$ requires nothing. If
$\amarking_2$ and $\amarking$ are connected and $\amarking$ is no
deadlock, there are also no requirements.  Otherwise $\amarking_2$ and
$\amarking$ are connected and $\amarking$ is a deadlock. In this case,
there must be a marking $\amarking_{\circlearrowleft}$ (valuation for
$\avariablebis_\ell$) so that formula $\aformulabis$ is true for
$(\amarking_1,\amarking_2,\amarking,\amarking_{\circlearrowleft})$. The
formula $\aformulabis$ checks that deadlock $\amarking$ is reachable
in both $\apn_1$ and $\apn_2$, see Figure~\ref{GRAPHE1}. Thus,
$\reach{\apn_1} = \reach{\apn_2}$ iff $\RG{\overline{\apn}}\models
\aformula$.  This proves the claimed undecidability. 
\end{proof}


\begin{openproblem}
Decidability status of $\MC{\FOF(\step{*})}{URG}$.
\end{openproblem}

While forward formulae can well identify the deadlock markings used in
the proof schema, the difficulty is in the description of the local
environment witnessing the simulation results.


\subsection{Taming undecidability with fragments}
\label{section-taming-undecidability}
In this section, we present the restrictions of $\FO(\step{})$ that we
found to have decidable model checking or validity problems.
\subsubsection{Existential  fragment}
Our undecidability results follow a common principle, namely
identifying a local pattern in the reachability graph that
characterizes an undecidable property.  The pattern may depend on the
specification language.
Below, we state a result that, at first glance, might seem to contradict
the previous findings: decidability of $\MC{\FO(\step{})}{URG}$ restricted to
the existential fragment. This decidability, however, simply implies that universal
quantification is needed to characterize undecidable properties by local
patterns.
We write $\exists \FO$ for the fragment of $\FO$ consisting of those
formulae that use only existential quantification when written in
prenex normal form.
\newcommand{\stmntpropdecexist}{
$\MC{\exists \FO(\step{},=)}{URG}$ 
is decidable.
}
\begin{proposition} \label{proposition-decidability-existential}
\stmntpropdecexist
\end{proposition}
\begin{proof}
Let $\apn =(P,T,F,\amarking_0)$ be a Petri net with reachability set
$\reach{\apn}$ and $|P| = n$. Decidability follows from two crucial
properties:\vspace{0.1cm}
\begin{enumerate}[(1)]
\item Given a Presburger formula $\aformula(\vect{\avariable}_1,
  \ldots, \vect{\avariable}_{\alpha})$ with $n \times \alpha$ free
  variables such that each $\vect{\avariable}_i$ is a sequence of $n$
  distinct variables interpreted as a marking of $N$, one can decide
  whether $\aformula(\amarking_1, \ldots, \amarking_{\alpha})$ holds
  true for some (not necessarily distinct) markings $\amarking_1,
  \ldots, \amarking_{\alpha}$ in $\reach{\apn}$.
  Proposition~\ref{proposition-reach-semilinear} corresponds to the
  case $\alpha = 1$.
\item One can effectively construct a quantifier-free Presburger
  formula $\aformula_{\step{}}(\vect{\avariable}_1,
  \vect{\avariable}_2)$ so that for all markings $\amarking_1,
  \amarking_2$, formula $\aformula_{\step{}}(\amarking_1,
  \amarking_2)$ holds iff $\amarking_1 [ \atransition \rangle
  \amarking_2$ for some $\atransition \in T$.
\end{enumerate}
\vspace{0.1cm}Before we turn to the proofs of (1) and (2), we explain how these
results yield decidability of $\MC{\exists \FO(\step{},=)}{URG}$. Consider $\aformulabis = \exists \
\avariable_1, \ldots, \avariable_{\alpha} \
\aformulabis'$ where $\aformulabis'$ is a quantifier-free formula with atomic
propositions of the form $\avariable_i \step{}  \avariable_{j}$ and
$\avariable_i
= \avariable_{j}$. With (2), one constructs a quantifier-free Presburger formula
$\aformula(\vect{\avariable}_1, \ldots, \vect{\avariable}_{\alpha})$
so that for all markings $\amarking_1, \ldots, \amarking_{\alpha}$ in
$\reach{\apn}$, formula
$\aformula(\amarking_1, \ldots, \amarking_{\alpha})$ holds true
iff
$\RG{\apn}, \avaluation \models \aformulabis'$ where
$\avaluation(\vect{\avariable}_i) = \amarking_i$
for $1\leq i\leq \alpha$.
By (1),  it is decidable whether
      $\aformula(\amarking_1, \ldots, \amarking_{\alpha})$ holds for some
markings
      $\amarking_1, \ldots, \amarking_{\alpha} \in \reach{\apn}$. This is
equivalent to $\URG{\apn} \models \aformulabis$.

It remains to prove (1) and (2). The formula $\aformula_{\step{}}(\vect{x}_1,
\vect{x}_2)$ for statement (2) encodes the definition of enabledness and firing
for transitions, $\amarking [\atransition \rangle \amarking'$:
$$
\bigvee_{\atransition \in T}
(\bigwedge_{\aplace \in P} \vect{\avariable}_1(p) \geq F(p,\atransition))
\wedge
(\bigwedge_{\aplace \in P} \vect{\avariable}_2(p) = \vect{\avariable}_1(p) -
F(p,\atransition) + F(\atransition,p)).
$$

For statement (1), we adapt the proof of
Proposition~\ref{proposition-reach-semilinear}.
We construct a Petri net $\apn'$ that simulates $\alpha$ copies of
$\apn$.  Technically, $\apn'$ is defined as the disjoint union of
$\alpha$ instances of $\apn$. The initial marking of $\apn'$ is
$\alpha$ times $\amarking_0$. For all markings $\amarking_1, \ldots,
\amarking_{\alpha}$ we now have the following equivalence: the
markings are reachable in $\apn$ and satisfy $\aformula(\amarking_1,
\ldots, \amarking_{\alpha})$ iff $(\amarking_1, \ldots,
\amarking_{\alpha})$ is a possible simulation result in $\apn'$ and
$\aformula(\amarking_1, \ldots, \amarking_{\alpha})$ holds.  An
application of Proposition~\ref{proposition-reach-semilinear} on
$\apn'$ and $\aformula$ yields the desired decidability result.
\end{proof}


Again, decidability is preserved with Presburger-definable properties on markings and
with labelled transition relations of the form $\step{\atransition}$.
\begin{corollary}
  $\MC{\FO(\step{},=)}{URG}$ restricted to Boolean combinations of
  existential formulae is decidable.
\end{corollary} 
\noindent Consequently, the following subgraph isomorphism problem is decidable
too:

\vspace{0.1cm}
\begin{description}
\item[\quad input] a finite directed graph $\mathcal{G} = \pair{V}{E}$ and a
Petri net
             $\apn$.
\item[\quad question] is there a subgraph of $\pair{\reach{N}}{\step{}}$ 
isomorphic to 
$\mathcal{G}$?\vspace{0.1cm}
\end{description}
\begin{openproblem} Decidability status of $\MC{\exists \FO(\step{*})}{URG}$
and $\MC{\exists \FO(\step{*},\step{})}{URG}$.
\end{openproblem} 
\subsubsection{$\ML(\bo{})$ with arithmetical constraints} 
\label{section-dec-val-ml}
Section~\ref{section-undecidability-ml} proves that
$\VAL{\ML(\bo{},\bo{}^{-1})}{URG}$ is undecidable. To our surprise,
and in contrast to the negative result on model checking the forward
fragment of $\FO$, this undecidability depends on the backward
modality. The following Proposition~\ref{proposition-dec-ml} shows
decidability of the validity problem for $\ML(\bo{})$, even in the
presence of arithmetical constraints at the atomic level.

\newcommand{\stmntpropdecml}{
  The validity problem $\VAL{\PAML(\bo{})}{URG}$ is decidable.
}
\begin{proposition} \label{proposition-dec-ml}
\stmntpropdecml%
\end{proposition}%
\begin{proof}
  Let $\apn$ be a Petri net, and $\aformula$ a formula in
  $\PAML(\bo{})$. According to Lemma~\ref{lemma-semi-paml} stated hereafter, 
  the set of
  markings satisfying $\neg \aformula$ is effectively semilinear. Let
  $\aset_{\neg \aformula}$ be this set.  Proving validity of $\aformula$
  amounts to checking that no element of $\aset_{\neg\aformula}$ is
  reachable in $\apn$. This is decidable from
  Proposition~\ref{proposition-reach-semilinear}.
  %
\end{proof}

\begin{lemma}\label{lemma-semi-paml}
Given a Petri net $\apn$ with $n$ places and a formula $\aformula$ in 
$\PAML(\bo{})$, the set of markings in $\Nat^n$ satisfying $\aformula$
in $\UG{\apn}$ is effectively semilinear.
\end{lemma}

\newcommand{\alltrans}{\step{}^{-1}}
\begin{proof}
  We proceed by induction on the structure of $\aformula$, using 
  the fact that semilinear sets are (effectively) closed under 
  Boolean operations
  and the fact that, if $\aset$ is semilinear, then $pre(\aset) = 
  \set{\amarking \in \Nat^n: \exists \ \amarking' \in \aset, \ \amarking \step{} \amarking'}$
  is effectively semilinear too.
  The latter set $pre(\aset)$ contains all markings with a successor marking in $\aset$.

  Each atomic formula  is a quantifier-free Presburger formula, and as
  such, defines a semilinear set. 
  Throughout the induction on the structure of $\aformula$, formulae with 
  outermost Boolean connectives
  are treated in the obvious way by applying Boolean operations on semilinear 
  sets.
  Eventually one has to prove
  that $\bo{}\aformulabis$ defines a semilinear set whenever
  $\aformulabis$ does. Using the induction hypothesis, let $\aset_{\aformulabis}$
  be the semilinear set of markings satisfying $\aformulabis$. The set 
  satisfying
  $\bo{}\aformulabis$ is then equal to
  $\Nat^n \setminus pre(\Nat^n \setminus \aset_{\aformulabis})$, which is effectively semilinear.
  This concludes the induction, and the proof.
  %
  %
\end{proof}



This decidability result can be extended by allowing labels on edges
(transitions).

\subsection{On the hardness of  decidable problems}
\label{section-hardness}
\noindent Some of our decision procedures call subroutines for checking
reachability in Petri nets, even though the reachability problem is not
known to be primitive recursive. We provide here some
complexity-theoretic justification for these costly invocations: we reduce the reachability problem for Petri nets
to the decidable problems $\MC{\ML(\bo{},\bo{}^{-1})}{URG}$ and
$\MC{\exists \FO(\step{})}{URG}$.
Besides reachability, we proposed decision procedures that
exploit the
effective
semilinearity of reachability sets or relations
(see e.g.
Proposition~\ref{proposition-dec-reach}).
The next proposition shows that, already
for bounded Petri
nets, $\MC{\FO(\step{})}{URG}$
is of high complexity.
\newcommand{\stmntpropcompbound}{
 $\MC{\FO(\step{})}{URG}$ restricted to bounded Petri nets
is decidable but this problem has nonprimitive recursive complexity.
}
\begin{proposition}\label{prop-complexity-bounded}
\stmntpropcompbound
\end{proposition}%
\begin{proof}
  We perform a reduction from the finite containment problem for Petri
  nets, known to have nonprimitive recursive
  complexity~\cite{Mayr&Meyer81}. Let $\apn_1$ and $\apn_2$ be two
  bounded Petri nets with identical sets of places, and construct
  $\overline{\apn}$ as in Section~\ref{section-undec-FO}.  This net is
  bounded. The formula $\aformula$ in $\FO(\step{})$ that checks
  inclusion is derived from the formula in
  Section~\ref{section-undec-FO}:
$$
\forall  \avariableter \
(\neg \exists  \avariableter'\;  \avariableter
\rightarrow  \avariableter')
 \Rightarrow (\exists  \avariableter_2\;
\avariableter_2\rightarrow  \avariableter
\wedge \neg \lsuc( \avariableter_2))
$$
where $\lsuc(\avariable) \egdef \exists \avariablebis \ ( \avariable
\step{} \avariablebis \wedge \avariablebis \step{}
\avariablebis)$. The construction guarantees $\reach{\apn_1} \subseteq
\reach{\apn_2}$ iff $\URG{\overline{\apn}} \models \aformula$. Indeed,
a deadlock is either reachable from $N_2$ or from $N_1$. But to
satisfy the formula, if the deadlock is reachable from $N_1$ it also
has to be reachable from $N_2$.  Note that the formula $\aformula$ is
again independent of $\apn_1$ and $\apn_2$.
\end{proof}



We have seen that $\VAL{\ML(\bo{}}{URG}$ is decidable by reduction to the
reachability problem for Petri nets (see Proposition~\ref{proposition-dec-ml}).
Below, we state that there is a reduction in the reverse direction, from
non-reachability to $\VAL{\ML(\bo{}}{URG}$.

\begin{proposition} There is a logarithmic-space reduction from
the non-reachability problem for Petri nets to $\VAL{\ML(\bo{})}{URG}$.
\end{proposition}
\begin{proof}%
  Without any loss of generality, we can assume that the
  non-reachability problem is restricted to the target marking
  $\vec{0}$ (no place has any token). Consider the Petri net $\apn =
  (P,T,F, \amarking_0)$ where we assume w.l.o.g. that every transition
  has a place in its preset. We build a variant Petri net $\apn'$ from
  $\apn$ by adding a new transition $t_p$ for every place $p\in
  P$. The new transitions are put in self-loop with their places,
  $F'(p, t_p)=1=F'(t_p, p)$ and $F'(p', t_p)=0=F'(t_p, p')$ for all
  $p'\in P$ with $p'\neq p$. Intuitively, $t_p$ witness for the
  presence of tokens on $p$ by the existence of at least one
  transition from $M$ in the reachability graph. As a result, $\vec{0}
  \not \in \reach{\apn}$ iff for every marking $\amarking \in
  \reach{\apn'}$, some transition can be fired:
  $\pair{D}{\step{}}, \amarking \models \diam{} \top$. Note that our
  reduction uses a constant formula.
\end{proof}

\newcommand{\stmntpropreducReachPNml}{
There is a logarithmic-space reduction from
the reachability problem for Petri nets to $\MC{\ML(\bo{},\bo{}^{-1})}{URG}$.
}%
\begin{proposition} \label{proposition-reduc-ReachPN-ml}
\stmntpropreducReachPNml
\end{proposition}%
\begin{proof}%
We reduce reachability of marking $M_2$ from marking $M_1$ in a Petri net
$\apn$ to an instance of $\MC{\ML(\bo{},\bo{}^{-1})}{URG}$ for a larger net
$\overline{\apn}$. The idea is to introduce a marking $\amarking_w$ (see
Figure~\ref{Figure:ReachComplexity}) such that the existence of a path to
$\amarking_w$ of length greater than $1$ is a witness for the existence of some
path from $\amarking_1$ to $\amarking_2$ in $\RG{\apn}$. To reach $\amarking_w$ by an $\ML$
formula, we place it close to the new initial marking. We sketch the
argumentation. The
initial marking $\amarking_0$ of $\overline{\apn}$ contains a single marked place $p_i$
for which two transitions $t_{try}$ and $t_{0}$ compete. Transition $t_{try}$
moves the unique token from $p_i$ to another place $p_{w}$ and thus produces the
marking $\amarking_w$ {\em where no other place is marked}. Transition  $t_{0}$ loads
$\amarking_1$ in the places of $\apn$ and moves the control token from $p_i$ to another
control place $p_{c}$ set in self-loop with all transitions of $\apn$. This
starts the simulation of $\apn$ from $\amarking_1$. The simulation may
get stuck or proceed forever, or it may
be interrupted whenever it reaches a marking of $\apn$ greater than or equal to
$\amarking_2$. Then, transition $t_{stop}$ consumes $\amarking_2$ from the
places of $\apn$ and moves the control token from $p_{c}$ to a place $p_{w'}$.
The control token is finally moved from $p_{w'}$ to $p_w$ by firing $t_{win}$.
$\amarking_w$ is reached, after firing
$t_{stop}\,t_{win}$, iff $\overline{\amarking_2}$ is reached.
Therefore $M_2$ is reachable from $\amarking_1$ iff
$\amarking_w$ is reachable from $\overline{\amarking_1}$
(its restriction to the places of $\apn$ is $\amarking_1$). This is equivalent to stating that $M_w$ has a
predecessor different from $\amarking_0$. The shape of the reachability graph allows us to formulate the latter as a local property in $\ML(\bo{},\bo{}^{-1})$:
\begin{align*}
\aformula\ :=\ \diam{}(\bo{}\bottom\wedge \ \diam{}^{-1} \diam{}^{-1} \top).
\end{align*}
Without loss of generality,
we can assume that $\amarking_1$ is no deadlock
and $\amarking_2\neq \amarking_1$. Formula $\aformula$ requires that
$\amarking_0$ has a deadlock successor which has an incoming path of
length two. That the successor is a deadlock means it is not $\overline{\amarking_1}$
but
$\amarking_w$ obtained by firing $t_{try}$. The path from $\amarking_0$ to
$\amarking_w$ is of length one and $\amarking_0$ has no predecessor. So
the path of length two to $\amarking_w$ is not via $t_{try}$ but stems from
$t_{win}$. This means $\amarking_w$ is reachable from $\overline{\amarking_1}$, which
means $\amarking_2$ is reachable from $\amarking_1$ in $\apn$.
\end{proof}
\begin{figure}
\begin{center}
  {\scriptsize
  %
  \begin{tikzpicture}[node distance=1cm,scale=1.3,transform shape]
    \tikzstyle{every transition}=[fill, minimum width=5mm, minimum height=2mm];
    \tikzstyle{arc}=[arrows=-latex,auto=left];
    \tikzstyle{multiarc}=[arrows = -latex, bend left, thick, shorten >=2pt];
    \tikzstyle{subnet}=[draw,dashed,shape=regular polygon, regular polygon sides=6];
    \newcommand{\mysize}{1.5cm}
    
    \draw
    node (zero){}
    node[subnet, minimum width=\mysize, minimum height=\mysize](N){
      $\apn$
    }
    node[node distance=1mm,right=of zero](zeroprime){}
    node [node distance=1mm, above left=of zeroprime](m1){$\overline{\amarking_1}$}
    node [node distance=1mm, below left=of zeroprime](m2){$\overline{\amarking_2}$}
    node[node distance=3cm,left=of m1](m0){$\amarking_0$}
    node[node distance=.7cm,left=of m2](m2p){$\amarking'_w$}
    node[node distance=.7cm,left=of m2p](mw){$\amarking_w$}
    (m0) edge [arc] node {$t_0$} (m1)
    (m2) edge [arc] node {$t_{stop}$} (m2p)
    (m0) edge [arc,swap] node {$t_{try}$} (mw)
    (m2p) edge[arc] node {$t_{win}$}(mw)
    ;
\end{tikzpicture}
}

\end{center}
\caption{Reachability graph in the hardness proof of
$\ML(\bo{},\bo{}^{-1})$-model
checking
\label{Figure:ReachComplexity}}
\end{figure}
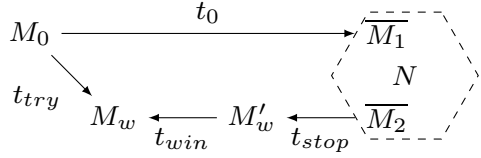
The proof of Proposition~\ref{proposition-reduc-ReachPN-ml} can be adapted
to $\exists \FO(\step{})$ for which we also have shown decidability of
model-checking by reduction to the reachability problem for Petri nets.
\newcommand{\stmntpropcompFOex}{
 There is a logarithmic-space reduction from
the reachability problem for Petri nets to $\MC{\exists \FO(\step{})}{URG}$.
}

\begin{proposition}
\label{proposition-complexity-FOone}%
There is a logarithmic-space reduction from the reachability problem
for Petri nets to $\MC{\exists \FO(\step{})}{URG}$ restricted to a
single variable.
\end{proposition}

\begin{proof} Among any two of the following problems, there is a
logarithmic-space reduction:
\begin{enumerate}[(1)]
\item the reachability problem for Petri nets;
\item the reachability problem for Petri nets restricted to instances such that
      the target marking $\amarking$ is equal to $\vec{0}$;
\item the following variant of the reachability problem:
      \begin{description}
      \itemsep 0 cm
      \item[input] a Petri net  $\apn = (P,T,F, \amarking_0)$ with no neutral transitions
      and a place $\aplace \in P$.
      \item[question] Is there a marking $\amarking$ with $\amarking(\aplace) = 0$
                       such that $\amarking \in \reach{\apn}$?
      \end{description}
\end{enumerate}
To show that (3) is as hard as reachability, the idea is to introduce a budget place that maintains the sum of tokens in all other places.
From an instance of problem (3), let us build in instance of
 $\MC{\exists \FO(\step{})}{URG}$ restricted to a single variable.
We build a Petri net $\apn'$ from $\apn$ and $\aplace \in P$ by simply adding
a neutral transition (the unique one in $\apn'$) that is in self-loop with $\aplace$.
One can then easily show that
there is a marking $\amarking \in \reach{\apn}$ with $\amarking(\aplace) = 0$
iff $\RG{\apn'} \models \exists  \avariable \  \neg (\avariable \step{}
\avariable)$.
\end{proof}

\section{$\FO$ with Reachability Predicates}
\label{section-reachability-predicates}

In this section, we consider several first-order languages with
reachability relations $\step{*}$ or $\step{+}$, mainly without the
one-step relation $\step{}$.  Undecidability of these dialects does
not directly follow from Theorem~\ref{prop:undecFO3} since we may
exclude $\step{}$.  Nonetheless we follow the same proof schema.
Besides, we distinguish the case when reachability sets are semilinear
leading to a surprising undecidability result
(Proposition~\ref{proposition-undec-acc-semilinear-sets}).  Finally,
we show that model-checking unlabelled graphs with
$\FO(\step{},\step{*})$ is undecidable too.

\subsection{$\FO$ with reachability relations}
\label{section-fo-reachability-relations}
Let us see why the model checking problem for both the strict and the non-strict reachability relation is undecidable.
\subsubsection{Undecidability of  $\MC{\FO(\step{+})}{URG}$}

The decidability status of $\MC{\FO(\step{+})}{URG}$ is not directly
dependent upon the decidability status of $\MC{\FO(\step{})}{URG}$.
Still we are able to adapt the construction of
Section~\ref{section-undec-FO} but using now a formula $\aformula$ in
$\FO(\step{+})$.  The Petri net $\overline{\apn}$ is the one depicted
on Figure~\ref{GRAPHE1}.
\newcommand{\dead}{dl}
\newcommand{\finaloop}{sl}
\newcommand{\predead}{predl}
The formula $\aformula$ is defined as follows:
$$
\aformula\egdef
\forall \ \avariableter\; \dead(\avariableter) \ \Rightarrow \
(\exists \ \avariableter_1\; (\avariableter_1 \step{+} \avariableter) \wedge
\aformula_{left}(\avariableter_1)) \wedge (\exists \ \avariableter_2\;  (\avariableter_2 \step{+} \avariableter)
 \wedge \aformula_{right}(\avariableter_2))
$$
where
\begin{iteMize}{$\star$}
\item $\dead(\avariableter) \egdef \neg \exists \avariableter'\; \avariableter \step{+} \avariableter'$,
\item $\finaloop(\avariablebis) \egdef \avariablebis \step{+} \avariablebis
\ \wedge \ \forall {\sf w}\; [ \avariablebis \step{+} {\sf w} \Rightarrow
 {\sf w} \step{+} \avariablebis]$,
\item $\aformula_{left}(\avariableter) \egdef [\exists \ \avariablebis  \;
\avariableter \step{+} \avariablebis \wedge \finaloop(\avariablebis)] \wedge
\ [\forall \avariablebis \; \avariableter \step{+} \avariablebis
\Rightarrow (\finaloop(\avariablebis) \ \vee \ \dead(\avariablebis))]$,
\item $\aformula_{right}(\avariableter) \egdef [\exists \avariablebis \; \avariableter
\step{+} \avariablebis
\wedge \forall \avariablebis \; \avariableter \step{+} \avariablebis \Rightarrow \dead(\avariablebis)]$.
\end{iteMize}
\newcommand{\stmntlemundecreachstrict}{
  $\reach{\apn_1}  = \reach{\apn_2}$
  iff
  $\RG{\overline{\apn}} \models \aformula$.
}
\begin{lemma}\label{lemma-undecidability-reach-strict}
\stmntlemundecreachstrict
\end{lemma}
\begin{proof}
  The principles presented in the proof of
  Lemma~\ref{lemma-undecidability-schema} apply here. Below, we
  refer to markings as they are depicted in Figure~\ref{GRAPHE1}.

  First, observe that none of the formulae $\dead(\avariableter),
  \finaloop(\avariablebis), \aformula_{left}(\avariableter)$ nor
  $\aformula_{right}(\avariableter)$ may be satisfied at a marking
  reached in course of simulating the original
  Petri nets
  $\apn_1$ or $\apn_2$: the formula $\dead(\avariableter)$,
  which asserts the absence of a successor, is always
  false on such markings whereas the formula $\finaloop(\avariablebis)$, requiring
  that one can always come back to $\avariablebis$, is false at such markings
  since the transitions $\atransition^1_{end}$ and $\atransition^2_{end}$ cannot be undone.
  Furthermore, neither $\dead$ nor $\finaloop$
  is satisfied
  by the markings $\amarking_1$ or $\amarking_2$. Hence,
  formulae $\aformula_{left}(\avariableter)$ and
  $\aformula_{right}(\avariableter)$ are not satisfied by any marking
  $\avariableter$
  reached in the course of simulating
  $\apn_1$ or $\apn_2$: any such
  marking has at least one successor of the type $\amarking_1$
  or $\amarking_2$, thus invalidating the subformulae
  $\forall \avariablebis \; \avariableter \step{+} \avariablebis
\Rightarrow (\finaloop(\avariablebis) \ \vee \ \dead(\avariablebis))$
and
  $\forall \avariablebis \; \avariableter \step{+} \avariablebis \Rightarrow \dead(\avariablebis)$.

  Now, it is straightforward to verify the following facts:
  \begin{iteMize}{$\star$}
  \item $\dead(\avariableter)$ is satisfied precisely at markings
    $\amarking_r$ and $\amarking_\ell$;
  \item $\finaloop(\avariablebis)$  is satisfied precisely at marking $\amarking_\circlearrowleft$;
  \item $\aformula_{left}$ and $\aformula_{right}$ are satisfied respectively
    at markings $\amarking_1$ and $\amarking_2$.
  \end{iteMize}
  The formula $\aformula$ 
  may be written $\forall  \avariableter\;
  \aformula'(\avariableter)$ with $\aformula'(\avariableter)$
  of the form
  $\dead(\avariableter) \
  \Rightarrow \ \aformulabis(\avariableter)$. Formula $\aformula'(\avariableter)$ is true
  whenever $\avariableter$ evaluates to a non-deadlock marking. Otherwise,
  when $\avariableter$ is a deadlock, validity of
  %
  $\aformulabis$ requires that it has two distinct predecessors $\avariableter_1$ and $\avariableter_2$
  of the
  types $\amarking_1$ and $\amarking_2$, entailing the
  equality of the reachability sets of $\apn_1$ and $\apn_2$.
  Conversely, if both reachability sets are equal, then all markings
  of $\apn_1$ and $\apn_2$ are connected 
  as described in
  Figure~\ref{GRAPHE1}, entailing the validity of $\aformula$ in
  $\overline{\apn}$.  
\end{proof}

\newcommand{\stmntundecfotrans}{
  $\MC{\FO(\step{+})}{URG}$ is undecidable. Furthermore this results
  holds for the fixed formula $\aformula$ defined earlier.
}
\begin{corollary} \label{corollary-undecibility-FOtransitive}
  \stmntundecfotrans
\end{corollary}
\subsubsection{Undecidability of  $\MC{\FO(\step{*})}{URG}$}

For showing undecidability of $\MC{\FO(\step{*})}{URG}$, we have
to adapt our usual proof schema since, in $\FO(\step{*})$,
we are no longer able to identify \sloops as we did in
$\FO(\step{+})$.
The new schema is illustrated in
Figure~\ref{GRAPHE2}.

\newcommand{\stmntpropundecreachnonstrict}{
  $\MC{\FO(\step{*})}{URG}$ is undecidable.
}

\begin{proposition} \label{prop-undec-reach-non-strict}
\stmntpropundecreachnonstrict
\end{proposition}

\begin{figure}[htbp]
\begin{center}
{\small 
  %
  \begin{tikzpicture}[scale=.8]
    \tikzstyle{every transition}=[fill, minimum width=5mm, minimum height=2mm];
    \tikzstyle{multiarc}=[arrows = -latex, bend left, thick, shorten >=2pt];
    \tikzstyle{subnet}=[draw,dashed,shape=regular polygon, regular polygon sides=6];
    \tikzstyle{arc}=[arrows=-latex,auto=left];

    \draw
    node[subnet](N1){$\apn_1$}
    node[subnet, node distance=3cm, right=of N1](N2){$\apn_2$} 
    node [node distance=1.4cm, right=of N1](tempo){}
    node [node distance=.5cm, above=of tempo](init){$\amarking_0$} 
    (init) edge [arc,swap] node {$\atransition^1_c$} (N1)
    (init) edge [arc] node {$\atransition^2_c$} (N2)
    node[below=of N1](z1){$\amarking_1$}
    node[below=of N2](z2){$\amarking_2$}
    node[below right=of z1](zl){$\amarking_\ell$}
    node[below left=of z2](zr){$\amarking_r$}
    (N1) edge [dashed, arrows=-latex] (z1)
    (N2) edge [dashed, arrows=-latex] (z2)
    (z1) edge [arrows=-latex] (zl)
    (z2) edge [arrows=-latex] (zr)
    ;
\end{tikzpicture}
}

\end{center}
\caption{Petri net $\overline{\apn}$ adapted for $\FO(\step{*})$}\label{GRAPHE2}
\end{figure}
\begin{proof}
  %
  From two Petri nets
  $\apn_1$ and $\apn_2$, we construct the Petri net
  $\overline{\apn}$ depicted in Figure~\ref{GRAPHE2}. We define the
  following formulae:
\begin{iteMize}{$\star$}
\itemsep 0 cm
\item  $\dead(\avariableter) \egdef  \forall {\sf w} \;  \avariableter \step{*}
       {\sf w} \Rightarrow  {\sf w}  \step{*}  \avariableter$,
\item $\predead(\avariableter) \egdef \neg \dead(\avariableter) \wedge
(\forall  {\sf w}  (\avariableter \step{*}
       {\sf w} \wedge \neg  {\sf w}  \step{*}  \avariableter) \Rightarrow \dead({\sf w}))$.
\end{iteMize}
Thus in Figure~\ref{GRAPHE2}, the markings $\amarking_r$ and
$\amarking_\ell$ satisfy $\dead$, and the markings
$\amarking_1$ and $\amarking_2$ satisfy $\predead$, but no other
marking satisfies these predicates.

The formula $\aformula$ is defined as follows:
$$
\aformula\egdef
\forall \avariableter \;(\dead(\avariableter)\Rightarrow \exists \avariableter_1,\avariableter_2\;
(\avariableter_1 \step{*} \avariableter \wedge
\predead(\avariableter_1)\wedge \avariableter_2 \step{*} \avariableter
\wedge \predead(\avariableter_2)\wedge \neg \avariableter_1 \step{*}
\avariableter_2))$$
Observe that $\neg \avariableter_1 \step{*} \avariableter_2$ ensures that
$\avariableter_1$ and $\avariableter_2$ have distinct interpretations.
%
By construction,
$\reach{\apn_1} = \reach{\apn_2}$ iff $\RG{\overline{\apn}} \models \aformula$.
\end{proof}

Even though $\MC{\FO(\step{}, =)}{UG}$ is decidable
(see Proposition~\ref{proposition-dec-ug}),
replacing $\step{}$ by $\step{*}$ and adding $init$ leads to undecidability.

\begin{corollary}\label{coro:undecFO_UG_init_reach}
$\MC{\FO(init,\step{*})}{UG}$ is undecidable.
\end{corollary}

Indeed, $\MC{\FO(\step{*})}{URG}$ reduces to
 $\MC{\FO(init,\step{*})}{UG}$ by relativization:
$\URG{\apn} \models \aformula$ iff
$\UG{\apn} \models \exists \avariable_0\ \ init(\avariable_0) \wedge \amap(\aformula)$
where $\aformula$ and $\amap(\aformula)$ are in $\FO(\step{*})$,
$\amap$ is homomorphic for Boolean connectives and
$\amap(\forall  \avariable \ \aformulabis) \egdef
\forall  \avariable \ (\avariable_0 \step{*} \avariable) \Rightarrow \amap(\aformulabis)$.

\medskip

\begin{openproblem} 
Decidability status of $\MC{\FO(\step{*})}{UG}$.
\end{openproblem}

\subsection{When semilinearity enters into the play}
We saw that $\MC{\FO(\step{},=)}{URG}$ restricted to Petri nets with
effectively semilinear reachability sets is decidable, using a
translation into Presburger arithmetic (see
Proposition~\ref{proposition-dec-reach}).  This section is devoted to
discovering what happens when the relation $\step{*}$ is added.  We
establish that $\MC{\FO(\step{},\step{*})}{URG}$ restricted to Petri
nets with semilinear reachability sets is undecidable, by a reduction
from $\MC{\FO(\step{})}{URG}$.  Given a Petri net $\apn$ and a
sentence $\aformula \in \FO(\step{})$, we reduce the truth of
$\aformula$ in $\RG{\apn}$ to the truth of a formula
$\overline{\aformula}$ in $\RG{\overline{\apn}}$ where
$\overline{\apn}$ is an augmented Petri net with a semilinear
reachability set.  The Petri net $\overline{\apn}$ is defined from
$\apn$ by adding the new places $\aplace_0$, $\aplace_1$ and
$\aplace_2$; each transition from $\apn$ is \insl with
$\aplace_1$. Moreover, we add a new set of transitions \insl with
$\aplace_2$, each of which adds tokens to or removes tokens from a
corresponding (original) place of $\apn$ (thus modifying its contents
arbitrarily).
These transitions form a subnet denoted by $\brownian$.
Three other transitions are added; see Figure~\ref{PN}
for a schematic representation of $\overline{\apn}$ (the initial marking $\amarking_0'$ of
$\overline{\apn}$ restricted to places in $\apn$ is $\amarking_0$, while
$\amarking'_0(\aplace_0) = \amarking'_0(\aplace_1) = 1$ and
$\amarking'_0(\aplace_2) = 0$).
%
%
Our intention is to force $\reach{\overline{\apn}}$ to be semilinear
while staying able to identify a subset from $\reach{\overline{\apn}}$
in bijection with $\reach{\apn}$; this is a way to drown $\reach{\apn}$
into $\reach{\overline{\apn}}$. Indeed, $\reach{\overline{\apn}}$ contains
all markings such that the sum of $\aplace_1$ and $\aplace_2$ is
$1$ and $\aplace_0$ is at most $1$.
Nevertheless, if the transition $\atransition$ is fired first, then the
subsequently
reachable markings are exactly those of $\apn$ (except that $\aplace_1$ contains one token); $\RG{\apn}$ embeds
isomorphically into $\RG{\overline{\apn}}$. Until $\atransition$ is
fired, one may always come back to $\amarking'_0$, using the brownian
subnet $\brownian$, but this is impossible afterwards.

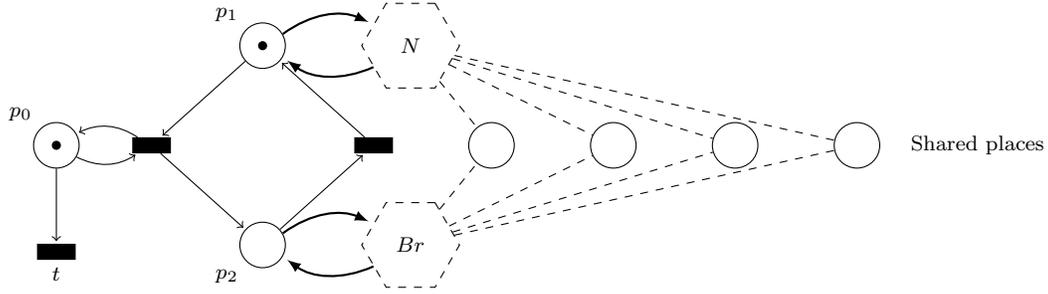
\begin{figure}[htbp]
\begin{center}
{\scriptsize
  %
  \begin{tikzpicture}[scale=.4]
    \tikzstyle{every transition}=[fill, minimum width=5mm, minimum height=2mm];
    \tikzstyle{multiarc}=[arrows = -latex, bend left, thick, shorten >=2pt];
    \tikzstyle{subnet}=[draw,dashed,shape=regular polygon, regular polygon sides=6, minimum size=13mm];
 
    \draw
    node[place,label=above left:$p_0$] (p0){}
    [children are tokens]
    child{ node [token] {} }
    node[right=of p0](pprime){}
    node[place,above right=of pprime,label=above left:$p_1$](p1){}
    [children are tokens]
    child{ node [token] {} }
    node[place,below right=of pprime,label=below left:$p_2$](p2){}
    node[transition, below  left=of p1](t1){}
    node[transition, below  right=of p1](t2){}
    node[transition, below=of p0,label=below:$t$](t){}
    (t)edge [pre] (p0)
    (t1)edge [pre,bend left] (p0)
    (t1)edge [pre] (p1)
    (t2)edge [pre] (p2)
    (t1)edge [post,bend right] (p0)
    (t1)edge [post] (p2)
    (t2)edge [post] (p1)
    node[subnet, right=of p1](N){$\apn$}
    node[subnet, right=of p2](Br){$Br$}
    (p1) edge [multiarc] (N)
    (p2) edge [multiarc] (Br)
    (N) edge [multiarc] (p1)
    (Br) edge [multiarc] (p2)
    \foreach \cur/\prev in {1/t2,2/pp1,3/pp2,4/pp3}
    {
      node [place, right=of \prev](pp\cur){}
      (N) edge[dashed] (pp\cur)
      (Br) edge[dashed] (pp\cur)
    }
    node [node distance=.3cm,right=of pp4]{Shared places}
    ;
  \end{tikzpicture}
}
%
\end{center}
\caption{Petri net $\overline{\apn}$}\label{PN}
\end{figure}
\newcommand{\stmntpropundecaccsem}{
$\MC{\FO(\step{},\step{*})}{URG}$ restricted to Petri nets with semilinear
reachability sets is undecidable.
}
\begin{proposition} \label{proposition-undec-acc-semilinear-sets}
\stmntpropundecaccsem
\end{proposition}
\begin{proof}
  In a first stage, we  use $init$ although this 
  predicate cannot be expressed in $\FO(\step{},\step{*})$. 
  Let $\overline{\aformula}$ be the formula $\exists \ \avariable_0 \
  \avariable_1 \ init(\avariable_0) \wedge \avariable_0 \step{} \avariable_1
  \wedge \neg (\avariable_1 \step{*} \avariable_0) \wedge
  \amap(\aformula)$ where $\amap(\cdot)$ is homomorphic for Boolean
  connectives
  and $\amap(\forall  \avariable \ \aformulabis) \egdef \forall 
  \avariable \ (\avariable_1 \step{*} \avariable) \Rightarrow
  \amap(\aformulabis)$ (relativization).  In $\overline{\aformula}$,
  $\avariable_0$ is interpreted as the initial marking $\amarking_0'$, 
  and $\avariable_1$ is 
  interpreted as a
  successor of $\avariable_0$ from which $\avariable_0$ cannot be reached 
  again. This may only
  happen  by firing $\atransition$ from $\amarking_0'$.  
  Now the relativization of every other variable to
  $\avariable_1$ in $\overline{\aformula}$
  ensures that 
  $\RG{\apn} \models \aformula$ iff  $\RG{\overline{\apn}} \models \overline{\aformula}$.
To remove $init$, we construct a 
Petri net 
$\overline{\apn}'$
  very similar to $\overline{\apn}$. $\overline{\apn}'$ has an extra place $p'_0$, initially marked with one
  token, and a new transition that consumes this token and produces two
  tokens in $p_0$ and $p_1$, which were initially empty. 
  By construction, the initial marking of $\overline{\apn}'$ is the sole
  marking in $\RG{\overline{\apn}'}$ 
  with no incoming edge and one outgoing edge.
  With this modified net, we use the modified formula
  $\overline{\aformula}'$ as follows:
  $$
  \exists \ \avariable_0' \ \avariable_0 \ \avariable_1 \ (\neg \exists
  \ \avariablebis \ \avariablebis \step{} \avariable_0') \ \wedge \
  \avariable_0' \step{} \avariable_0 \ \wedge \ \avariable_0 \step{} 
  \avariable_1
  \wedge (\neg \avariable_1 \step{*} \avariable_0) \wedge
  \amap(\aformula)
  $$
  For the same reasons as above, $\RG{\apn} \models \aformula$ iff
  $\RG{\overline{\apn'}} \models \overline{\aformula'}$.
\end{proof}
\begin{openproblem}
Decidability status of $\MC{\FO(\step{*})}{URG}$ restricted to Petri nets with semilinear
reachability sets.
\end{openproblem}
\subsection{The reachability relation and structure $\UG{\apn}$}
Corollary~\ref{coro:undecFO_UG_init_reach} has stated a first undecidability
result for the structure $\UG{\apn}$. In this section,
we examine two other situations where it is an undecidable problem to model-check
formulas of $\FO(\step{},\step{*})$ in $\UG{\apn}$.
\newcommand{\stmntpropfostepreach}{
  $\MC{\FO(\step{},\step{*})}{UG}$ is undecidable.
}
\begin{proposition}\label{prop-fo-step-reach}
\stmntpropfostepreach
\end{proposition}
%
\begin{figure}[htbp]
\begin{center}
{\small 
  %
  \begin{tikzpicture}[scale=.8]
    \tikzstyle{every transition}=[fill, minimum width=5mm, minimum height=2mm];
    \tikzstyle{multiarc}=[arrows = -latex, bend left, thick, shorten >=2pt];
    \tikzstyle{subnet}=[draw,dashed,shape=regular polygon, regular polygon sides=6];
    \tikzstyle{miniplace}=[place, minimum width=3mm, minimum height=3mm];
    \tikzstyle{minitrans}=[transition, minimum width=2.5mm, minimum height=1.2mm];
 
    \draw
    node[place,label=above left:$p_0$] (p0){}
    [children are tokens]
    child{ node [token] {} }
    node[node distance=3cm,subnet, minimum width=3.6cm, minimum height=3.1cm, right=of p0,label=above left:$\apn'$](Np){}
    node[node distance=4.7cm, right=of p0](Nprep){}
    node[transition, below right=of p0, label=below:$t_0$](t0) {}
    edge [pre] (p0) 
    edge [post] (Np)
    node[transition, above right=of p0, label=above right:$t_{e}$](te) {}
    edge [pre] (p0) 
    node[node distance=.3cm, above=of Nprep](newB){}

    node[node distance=.6cm,miniplace,label=above:$p_\ell$, right=of newB] (pl){}

    node [node distance=1.2cm,subnet, left=of pl](N){$\apn$}
    edge [pre, bend left=15] (pl) 
    edge [post,bend right=15] (pl)

    node [minitrans, below=of pl, label=below:$t_\ell$](tl){}
    edge [pre, bend left=15] (pl) 
    edge [post,bend right=15] (pl)

    node [minitrans, below=of N, label=right:$t_i$](t_i){}
    edge [pre, bend left=15] (N) 
    edge [post,bend right=15] (N)

    node [right =of pl](txt1){
      \begin{minipage}{0.3\linewidth}
        $p_\ell$: new place \insl with each transition of $\apn$.
      \end{minipage}
    }
    node [right =of tl](txt1){
      \begin{minipage}{0.3\linewidth}
        For each place $p_i$ in $\apn$, there is a transition $t_i$
        \insl with it.
      \end{minipage}
    }
    %
    ;
  \end{tikzpicture}
}

\end{center}
\caption{Petri net $\overline{\apn}$}
\label{fig:UG}
\end{figure}
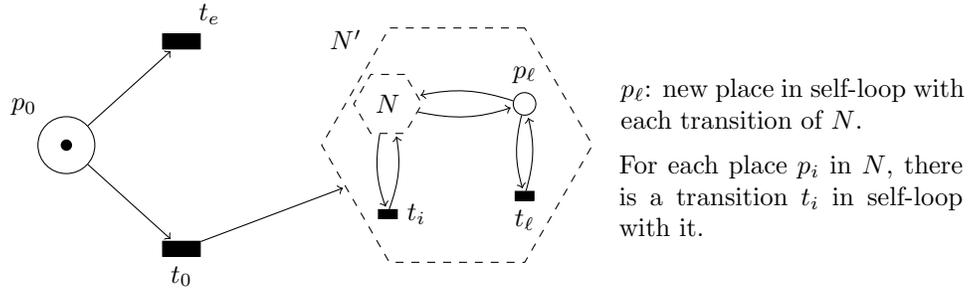
\begin{proof}
We reduce $\MC{\FO(\step{*})}{URG}$ to $\MC{\FO(\step{},\step{*})}{UG}$.
Given a net $\apn = (P,T,F, \amarking_0)$ and a formula $\aformula$ in 
$\FO(\step{*})$, we construct $\overline{\apn}$ and $\aformula'$ such that
$\URG{\apn} \models \aformula$ iff 
$\UG{\overline{\apn}} \models \aformula'$. 
 Figure~\ref{fig:UG}
  presents some key elements for the construction of
$\overline{\apn}$. 


  First, let $\apn' = (P',T',F', \amarking'_0)$ be the Petri net
  defined with $P' = P\cup \{p_\ell\}$, $T' = T\cup \{t_i\ |\ p_i\in
  P'\}$, for all $(p,t)$ in $P\times T$, $F' (p, t)= F(p,t)$ and $F'
  (t,p)= F(t,p)$, for all $p_i\in P', F(p_i,t_i)=F(t_i,p_i)=1$, for all
  $t\in T, F(t,p_\ell)=F(p_\ell,t)=1$, for all $p\in P$,
  $\amarking'_0(p) = \amarking_0(p)$, and $\amarking'_0(p_\ell) =1$.
  Restricted to places in $P$ (all places but $p_\ell$), the reachable
  markings of $\apn'$ coincide with those of $\apn$. 
  By construction, $p_\ell$ contains always a single token. In
  $\URG{\apn'}$, every marking has a \sloop.
  Similarly, every marking of $\apn'$ in which some place is positive
  possesses a \sloop in the graph $UG(\apn')$. The tuple $(0,0,\ldots,0)$, on
  the other hand, enables no transition (the empty place $p_\ell$
  inhibits every transition).

  Now, we construct $\overline{\apn}$ from $\apn'$. $\overline{\apn}$
  has the same places and transitions as $\apn'$, plus an
  extra place $p_0$ and two extra transitions $t_e$ and
  $t_0$. Transition $t_e$ removes tokens from $p_0$, one at a time.
  Transition $t_0$ consumes one token from $p_0$ and produces $\amarking'_0$ in
  the places of $\apn'$. The initial marking
  $\overline{\amarking_0}$ of $\overline{\apn}$ has a
  single token in place $p_0$.

  We claim the following:
  \begin{iteMize}{$\star$}
  \item The reachable graph of $\overline{\apn}$ is identical to the
    reachable graph of $\apn$, up to the first transition and up to the \sloops which have no
    influence on formulas in $\FO(\step{*})$.
  \item There is a formula $\aformula_{init}(x)\in
    \FO(\step{},\step{*})$ which is satisfied in
    $\UG{\overline{\apn}}$ only at $\overline{\amarking_0}$.
  \end{iteMize}

  Assuming these claims, validity of a formula in $\FO(\step{*})$
  with respect to $\URG{\apn}$ may be reduced to the validity of a
  formula of $\FO(\step{},\step{*})$ with respect to
  $\UG{\overline{\apn}}$, using a similar technique as in the proof of
  Corollary~\ref{coro:undecFO_UG_init_reach}. 
  For this purpose, we should relativize the given formula in $\FO(\step{*})$
  to the vertices of $\UG{\overline{\apn}}$ that may be reached from the
  marking $\overline{\amarking_0}'$ defined by
  $\overline{\amarking_0}\stepbis{t_0}\overline{\amarking_0}'$. 
  This can actually be done in $\FO(\step{},\step{*})$, because 
  $\overline{\amarking_0}'$ is the sole marking of $\overline{\apn}$ that
  satisfies the formula $\exists \avariablebis\ 
  \aformula_{init}(\avariablebis)\wedge \avariablebis\step{}\avariable\wedge
  \avariable\step{}\avariable$. Therefore, to complete the proof of the
  proposition, it suffices to establish the two claims made above.

  
  Now, the first claim derives immediately from the construction of
  $\overline{\apn}$.
  The second claim may be established by setting:
  \[
  \aformula_{init} (\avariable) \egdef (\neg \ \avariable\step{}
  \avariable) \wedge (\exists \, \avariablebis \forall
  \,\avariableter\ \avariable\step{} \avariablebis \wedge\,
  \neg(\avariablebis\step{}\avariableter))
  \]
  This formula contains a subformula $(\neg \ \avariable\step{}
  \avariable)$ that expresses the absence of a \sloop, thus
  $\aformula_{init} (\avariable)$ may only be satisfied in markings
  with all places $\aplace \in P'$ empty. But $(\neg \
  \avariable\step{}\avariable)$ may be satisfied in a marking
  $\avariable$ with an arbitrary number of tokens in $\aplace_0$. Now
  consider markings with all places in $P'$ empty, and an arbitrary
  number of tokens in $\aplace_0$. Three cases must be
  considered. First, suppose that $\aplace_0$ contains a single token
  (i.e., $\avariable$ is interpreted by $\overline{\amarking_0}$),
  then $(\exists \, \avariablebis \forall\,\avariableter\
  \avariable\step{} \avariablebis
  \wedge\,\neg(\avariablebis\step{}\avariableter))$ is satisfied:
  $\avariable$ has a successor $\avariablebis$ (reached by firing
  $\atransition_e$) which is a deadlock. Second, if $\aplace_0$ is
  empty, then the marking $\avariable$ has no successor at
  all. If $\aplace_0$ contains at least two tokens, then no
    successor of $\avariable$ is a deadlock: every marking reached by
    $\atransition_0$ has a \sloop\ and $\atransition_e$ can be
    executed at least twice.  Putting everything together, the only
  tuple in $\Nat^{n}$ satisfying $\aformula_{init} (\avariable)$, is
  the marking $\overline{\amarking_0}=(1,0,\ldots,0)$, establishing
  the second claim.

  %
\end{proof}
%
Proposition~\ref{prop-fo-step-reach} holds even when the reachability set
of the net is effectively semilinear.
\newcommand{\stmntpropfostepreachsem}{
  $\MC{\FO(\step{},\step{*})}{UG}$ is undecidable for the subclass of Petri
  nets with an effective semilinear reachability set.  }
\begin{proposition}\label{prop-fo-step-reach-semi}
\stmntpropfostepreachsem
\end{proposition}
\begin{proof}
  %
We pile up (adaptations of) the proofs of Propositions~\ref{prop-undec-reach-non-strict}, \ref{proposition-undec-acc-semilinear-sets}, and \ref{prop-fo-step-reach}. 

Given arbitrary two nets $\apn_1$ and $\apn_2$ without \ntrts\hspace{-.5ex}, let $\apn_3$ denote the net $\overline{\apn}$ constructed from $\apn_1$ and $\apn_2$ like in the proof of Proposition~\ref{prop-undec-reach-non-strict}, and let $M_3$ denote the initial marking of this net. By the proof of Proposition~\ref{prop-undec-reach-non-strict}, $\reach{\apn_1}=\reach {\apn_2}$ if and only if $\RG{\apn_3} \models \aformula$, where:
$$
\aformula\egdef
\forall \avariableter \;(\dead(\avariableter)\Rightarrow \exists \avariableter_1,\avariableter_2\;
(\avariableter_1 \step{*} \avariableter \wedge
\predead(\avariableter_1)\wedge \avariableter_2 \step{*} \avariableter
\wedge \predead(\avariableter_2)\wedge \neg \avariableter_1 \step{*}
\avariableter_2 \wedge \neg \avariableter_2 \step{*} \avariableter_1)),$$
$$\dead(\avariableter) \egdef  \forall {\sf w} \;  \avariableter \step{*}  
       {\sf w} \Rightarrow  {\sf w}  \step{*}  \avariableter,$$ 
$$\predead(\avariableter) \egdef \neg \dead(\avariableter) \wedge
(\forall  {\sf w} ) (\avariableter \step{*}  
       {\sf w} \wedge \neg  {\sf w}  \step{*}  \avariableter) \Rightarrow \dead({\sf w}).$$
Let $\mbox{\RGV}_{\circlearrowleft}(\apn_3)$ be the extended reachability graph obtained from $\RG{\apn_3}$ by adding a \sloop in every marking. Then clearly, $\RG{\apn_3} \models \aformula$ if and only if $\mbox{\RGV}_{\circlearrowleft}(\apn_3) \models \aformula$. By Hack's result, $\mbox{\RGV}_{\circlearrowleft}(\apn_3) \models \aformula$ is undecidable from the input $\{\apn_1,\apn_2\}$.  

Now put $\apn=\apn_3$ in the net shown in Figure~\ref{PN}. Denote the resulting 
net $\overline{\apn}$ by $\apn_4$, and let $\amarking_4$ be 
its initial marking. By construction, $\apn_4$ has a semilinear reachability
set. Moreover, if we put:  
$$\theta(\avariable,\avariablebis)
\egdef \ init(\avariable) \wedge \avariable \step{} \avariablebis
\wedge \neg (\avariablebis \step{*} \avariable),$$
then, in $\RG{\apn_4}$, this statement holds exclusively for 
$\avariable$ interpreted by $\amarking_4$ and 
$\avariablebis$ interpreted by $\amarking_3+\{p_1\}$. Let $\mbox{\RGV}_{\circlearrowleft}(\apn_4)$ be the extended reachability graph obtained from $\RG{\apn_4}$ by adding a \sloop in every marking. Then clearly, in $\mbox{\RGV}_{\circlearrowleft}(\apn_4)$, 
$\theta(\avariable,\avariablebis)$ holds exclusively for $\avariable$ 
interpreted by $\amarking_4$ and $\avariablebis$ interpreted by 
$\amarking_3+ \{\aplace_1\}$. 

Finally put  $\apn=\apn_4$ in the net shown in Figure~\ref{fig:UG}. Denote the 
resulting net $\overline{\apn}$ by 
$\apn_5$, and let $\amarking_5$ be its initial marking. Thus, $\apn_5$ has a
semilinear reachability set. 
As was shown in the proof of Proposition~\ref{prop-fo-step-reach}, if we put:
$$
\aformula_{init} (\avariable) \egdef (\neg \ \avariable\step{}
\avariable) \wedge (\exists \, \avariablebis \forall
\,\avariableter\ \avariable\step{} \avariablebis \wedge\,
\neg(\avariablebis\step{}\avariableter)),
$$
then, in $\UG{\apn_5}$, $\aformula_{init}(\avariable)$ holds exclusively for 
$\avariable$ interpreted by $\amarking_5$. Therefore, if we put:
$$
\psi_{init} (\avariable) \egdef \exists \avariablebis\ 
\aformula_{init}(\avariablebis)\wedge \avariablebis\step{}\avariable\wedge
\avariable\step{}\avariable,
$$
then, in $\UG{\apn_5}$, $\psi_{init}(\avariable)$ holds exclusively 
for $\avariable$ interpreted by $\amarking_4+\{p_\ell\}$. 
The subgraph of $\UG{\apn_5}$ reachable from 
the marking $\amarking_4+\{p_\ell\}$ is isomorphic to 
$\mbox{\RGV}_{\circlearrowleft}(\apn_4)$. Therefore, in $\UG{\apn_5}$, $\psi_{init}(\avariable)\wedge\avariable\step{} \avariablebis \wedge \neg (\avariablebis \step{*} \avariable)$ holds for $\avariable,\avariablebis$ if and only if 
$\avariable$ is interpreted by $\amarking_4+\{p_\ell\}$ and 
$\avariablebis$ is interpreted by $\amarking_3+\{p_1\}+\{p_\ell\}$. 
The subgraph of $\UG{\apn_5}$ reachable from the marking $M_3+\{p_1\}+\{p_\ell\}$ is isomorphic to $\mbox{\RGV}_{\circlearrowleft}(\apn_3)$. Therefore, $\mbox{\RGV}_{\circlearrowleft}(\apn_3) \models \aformula$ if and only if $\UG{\apn_5}\models\overline{\aformula}$ where $\overline{\aformula}$ is the formula:
$$
 \exists \ \avariable_0 \ \avariable_1 \ \psi_{init}(\avariable_0)
\wedge \avariable_0 \step{} \avariable_1
\wedge \neg (\avariable_1 \step{*} \avariable_0) \wedge
\amap(\aformula)$$
where $\amap(\cdot)$ is homomorphic for Boolean connectives  
and $\amap(\forall  \avariable \ \aformulabis) \egdef \forall 
\avariable \ (\avariable_1 \step{*} \avariable) \Rightarrow
\amap(\aformulabis)$ (relativization).
As a consequence, $\UG{\apn_5}\models \overline{\aformula}$ is undecidable 
from the input $\{\apn_1,\apn_2\}$.  
 
\end{proof}
In this section we have examined several first-order sublanguages
involving the reachability predicate. We obtained undecidability
results, even when the reachable markings form a semilinear set, and
even when the global structure $\UG{\apn}$ is considered instead of
$\URG{\apn}$.


\section{Concluding Remarks}
\label{section-conclusion}%
We investigated mainly the model-checking problem over unlabelled
reachability graphs of Petri nets with the first-order language
$\FO(\step{})$ (no label on transitions, no property on markings).
The robustness of our main undecidability proof has been tested
against standard fragments of $\FO(\step{})$ (for instance the
two-variable fragment), modal fragments from $\ML(\bo{},\bo{}^{-1})$
and against the additional assumption that reachability sets are
effectively semilinear.  Table~\ref{figure-summary} provides a summary
of the main results (observe that whenever the reachability relation
$\step{*}$ is effectively semilinear, each problem is decidable).
Results in bold are proved in the paper, whereas unbold ones are their
consequences; furthermore each undecidability result holds for a fixed
formula.
\newcommand{\notesym}{\dagger}
\begin{table}
{\scriptsize 
\begin{center}
\begin{tabular}{c|c|c|c|}
  Problem& 
$\sharp$
& Arbitrary & Effectively semilinear $\reach{\apn}$ \\
  \hline
  $\MC{\FO(\step{})}{\sharp}$
  & $URG$ & {\bf \undec} \ (Cor.~\ref{prop:undecFO})   & {\bf \dec} \\
  \hline
  $\MC{\FO(\step{+})}{\sharp}$
  & $URG$& {\bf \undec} \  (Cor.~\ref{corollary-undecibility-FOtransitive}) & open  \\
  \hline
  $\MC{\FO(\step{*})}{\sharp}$  & $URG$& {\bf \undec} 
  \ (Prop.~\ref{prop-undec-reach-non-strict}) &  open \\
  \hline
  \multirow{2}{*}{$\MC{\FO(\step{},\step{*})}{\sharp}$}
  & $URG$& \undec & {\bf \undec} \ (Prop.~\ref{proposition-undec-acc-semilinear-sets}) \\
  & $UG$&  {\bf \undec} \ (Prop.~\ref{prop-fo-step-reach}) &  {\bf \undec} \ (Prop.~\ref{prop-fo-step-reach-semi})\\
  \hline
  $\MC{\FOP(\step{})}{\sharp}$ & $URG$ &  {\bf \undec} (Prop.~\ref{prop-positive}) &
   \dec \\
   \hline 
  $\MC{\FOF(\step{})}{\sharp}$ & $URG$ &  {\bf \undec} (Prop.~\ref{prop-forward}) &
   \dec \\
   \hline 
  $\MC{\exists \FO(\step{},=)}{\sharp}$ & $URG$& {\bf \dec}$^{\notesym}$ \ 
   (Prop.~\ref{proposition-decidability-existential}) & \dec \\
  \hline
  $\MC{\FO(\step{})}{\sharp}$ with 1 variable & $URG$& {\bf \dec}$^{\notesym}$ \ 
   (Prop.~\ref{proposition-complexity-FOone}) & \dec \\
  \hline
  $\MC{\FO(\step{},=)}{\sharp}$ & $UG$& {\bf \dec} \ (Prop.~\ref{proposition-dec-ug}) & \dec \\
  \hline
  $\MC{\ML(\bo{})}{\sharp}$ & $URG$ & {\bf \pspace-complete} &  \pspace-complete \\
  \hline
  $\MC{\ML(\bo{},\bo{}^{-1})}{\sharp}$ & $URG$ & {\bf \dec}$^{\notesym}$ \ 
  (Prop.~\ref{proposition-preliminaries-ml2}) & \dec \\
  \hline
  $\VAL{\ML(\bo{},\bo{}^{-1})}{\sharp}$ & $URG$ & {\bf \undec} \ (Prop.~\ref{proposition-undecidability-ml}) & \dec \\
  \hline
  $\VAL{\PAML(\bo{})}{\sharp}$ & $URG$ & {\bf \dec}$^{\notesym}$ \ (Prop.~\ref{proposition-dec-ml}) &  \dec \\
  \hline
\end{tabular}
\end{center}
}
\caption{Summary ($\notesym$: equivalent to Petri nets (non) reachability problem)}
\label{figure-summary}
\end{table}
We have investigated several types of borderlines to distinguish decidable problems from
undecidable ones. For instance, $\MC{\FO(\step{})}{URG}$ restricted to the two-variable fragment
is undecidable whereas $\MC{\FO(\step{})}{URG}$ restricted to the existential
fragment is decidable
(even though this problem is at least as hard as the reachability problem for Petri nets).
Similarly, on the modal side, $\MC{\ML(\bo{},\bo{}^{-1})}{URG}$ is decidable (again 
 as hard as the reachability problem for Petri nets) whereas 
 $\VAL{\ML(\bo{},\bo{}^{-1})}{URG}$ is undecidable. 
Despite the numerous results we obtained, we can identify the following
rules of thumb.
\begin{enumerate}[(1)]
\item Undecidability of $\MC{\FO(\step{})}{URG}$ is robust for numerous fragments of 
      $\FO(\step{})$ including both universal and
      existential quantifications (a single alternation is enough).
\item Decidability results with simple restrictions such as considering bounded Petri nets
      or $\exists \FO(\step{})$ lead to computationally 
      difficult problems, some of them
      being non primitive recursive or  as hard as the reachability 
      problem for Petri nets (see Section~\ref{section-hardness}).
\item The above points are still relevant for modal languages.
 \end{enumerate}
Let us conclude the paper by mentionning possible continuations of this
work. A first direction would be to investigate the model checking of 
fragments of {\em second-order languages} with respect to Petri net unlabelled
reachability graphs. Knowing that $\MC{\FO(\step{})}{URG}$ is already 
undecidable, this makes sense only if one disallows first-order 
quantification, while keeping of course second-order quantification.
A possible primitive atomic formula could be for instance: 
$X \SOimpl{} Y$ $\equivdef$ for all $x \in X$, there is $y
\in Y$ such that $x \step{} y$ and for all $y \in Y$, there is $x \in
X$ such that $x \step{} y$. With this definition, it is easily shown 
that $\MC{{\rm MSO}(\Rightarrow)}{URG}$ is undecidable, but many other 
fragments of MSO are worth investigating and comparing with the 
fragments considered in the paper. 

A second direction for extending this work would be to consider the
geometrical properties of the set of markings reachable from a given
marking, taken as a subset of $\Nat^n$. It is for instance trivial to
determine whether there is at least one marking reachable from the
initial marking and different from it. It is slightly more difficult
to prove that there is at least one non-reachable marking.

A third direction, diverging significantly from our approach, would 
be to investigate decidability questions about infinite 
{\em unfoldings} of nets instead of net reachability graphs. 
Unfolding Petri nets produces {\em local event structures} that 
induce in turn {\em local trace languages} \cite{HKT93}. Safe Petri 
nets, as opposed to unbounded Petri nets, may in particular be 
modelled with {\em regular} trace event structures \cite{Mad03}. 
The decidability of FO over regular trace event structures has been 
shown in \cite{Mad03}, as well as the decidability of MTL, a fragment
of MSO where quantification is restricted to conflict-free sets of
events. The proofs of these results rely strongly on regularity and 
do not extend easily to local event structures representing general 
Petri nets.


\bibliographystyle{abbrv}
\bibliography{biblio}
\vspace{-20 pt}

\end{document}